\newcommand\SANSCOMMENTAIRE[1]{#1}

\newcommand\shorter[2][]{}

\documentclass{article}

\title{A characterization of functions over the integers computable in polynomial time using discrete  differential equations~\thanks{This work was funded by the ANR - project Difference}
}
  
\author{Olivier Bournez\thanks{Laboratoire d'Informatique de l'X (LIX), CNRS, Ecole Polytechnique,  Institut
	  Polytechnique de Paris, 91128 Palaiseau
	 Cedex, France
}
\and Arnaud Durand\thanks{Université Paris Cité, CNRS, IMJ-PRG, Paris, France
}
}

\usepackage{lipsum}
\usepackage{amsfonts}
\usepackage{graphicx}
\usepackage{epstopdf}
\usepackage{algorithmic}
\usepackage{amsopn}
\usepackage{amsmath}
\usepackage{amsthm}
\usepackage{amssymb, amsfonts}
\usepackage{bm} 
\usepackage{url}
\usepackage{xifthen}
\usepackage{microtype}

\RequirePackage{a4wide}
\usepackage[utf8]{inputenc}
\usepackage[T1]{fontenc}
\newcommand\vectorl[1]{{\mathbf#1}}

\newcommand\vx{\vectorl{x}}

\newcommand\R{\mathbb{R}}
\newcommand\Q{\mathbb{Q}}
\newcommand\N{\mathbb{N}}
\newcommand\Z{\mathbb{Z}}

\newcommand\tZ{\tilde{\Z}}

\newcommand\Greg{Grzegorczyk}
\newcommand\E{\mathcal{E}}
\newcommand\REC{{\rm REC}}

\newcommand\SMIN{{\rm SMIN}}
\newcommand\limREC{{\rm BR}}
\newcommand\BRN{{\rm BRN}}
\newcommand\BSUM{{\rm BSUM}}
\newcommand\BPROD{{\rm BPROD}}
\newcommand\BSUMs{{\rm BSUM_{<}}}
\newcommand\BPRODs{{\rm BPROD_{<}}}
\newcommand\ODE{{\rm ODE}}

\newcommand\LI{{\rm LI}}
\newcommand\Elem{\mathcal{E}}
\newcommand\Gregn{\mathcal{E}_n}
\newcommand\PR{\mathcal{P}\mathcal{R}}

\newtheorem{theorem}{Theorem}[section]
\newtheorem{lemma}[theorem]{Lemma}
\newtheorem{corollary}[theorem]{Corollary}

\newtheorem{definition}[theorem]{Definition}

\newtheorem{example}[theorem]{Example}
\newtheorem{remark}[theorem]{Remark}

\def\propositions{proposition}

\newenvironment{oureqnarray}{\begin{equation}\begin{array}{lll}}{\end{array}\end{equation}}
\newenvironment{oureqnarrayd}{\begin{equation}\begin{array}{lllllll}}{\end{array}\end{equation}}
\newenvironment{oureqnarrayde}{\begin{equation*}\begin{array}{lllllll}}{\end{array}\end{equation*}}

\newcommand\dint[4]{\int_{#1}^{#2}{#3}{\delta #4}}

\newcommand\onehalf{\frac{1}{2}}

\newcommand\fallingexp[1]{\overline{2}^{#1}}

\usepackage{color}
\usepackage{todonotes}

\newcommand{\arnaud}[2][]{\SANSCOMMENTAIRE{\todo[inline,color=green!20,caption={2do}, #1]{\begin{minipage}{\textwidth-4pt}
Arnaud:			#2\end{minipage}}}}

\newcommand{\arnaudmargin}[1]{\SANSCOMMENTAIRE{\todo[color=green!40]{A: #1}}}
\newcommand{\olivier}[2][]{\SANSCOMMENTAIRE{\todo[inline,color=red!40,caption={2do}, #1]{\begin{minipage}{\textwidth-4pt}
Olivier:			#2\end{minipage}}}}

\newcommand{\sabrina}[2][]{\SANSCOMMENTAIRE{\todo[inline,color=blue!20!white,caption={2do},#1]{\begin{minipage}{\textwidth-4pt}
Sabrina:			#2\end{minipage}}}}

\newcommand{\fonction}[1]{\textsf{#1}}

\newcommand\projection[2]{\mathbf{\pi}_{#1}^{#2}}
\newcommand{\sucs}{\mathbf{s}}
\newcommand\myominus{\mathbf{\ominus}}
\newcommand\plus{\mathbf{+}}
\newcommand\minus{\mathbf{-}}
\newcommand\gE{\mathbf{E}}

\newcommand{\succun}[1]{\mathbf{1}({#1})}
\newcommand{\succzero}[1]{\mathbf{0}({#1})}

\def\moins{\mathrel{\dot{-}}}

\newcommand{\zero}{\mathbf{0}}

\newcommand{\sign}[1]{\fonction{sg}(#1)}

\newcommand{\signn}[1]{\fonction{sg}_{\N}(#1)}
\newcommand{\signcomp}[1]{\bar{\fonction{sg}}(#1)}
\newcommand{\signcompn}[1]{\bar{\fonction{sg}_\N}(#1)}

\newcommand{\logicalif}[3]{\fonction{if}(#1,#2,#3)}

\newcommand{\cond}[3]{\fonction{if}(#1,#2,#3)}
\newcommand{\condn}[3]{\fonction{if}_\N(#1,#2,#3)}

\newcommand{\case}{\fonction{case}}

\newcommand{\tu}[1]{\mathbf{#1}}

\newcommand{\dderiv}[2]{\frac{\partial #1}{\partial #2}}
\newcommand{\dderivL}[1]{\frac{\partial #1}{\partial \lengt}}
\newcommand{\dderivl}[1]{\frac{\partial #1}{\partial \ell}}

\newcommand{\lengt}{\mathcal{L}}
 \newcommand\lengthnotation{\ell}
\newcommand{\length}[1]{\mathrm{\lengthnotation}(#1)}
\newcommand{\degre}[1]{\mathrm{deg}(#1)}

\newcommand{\inst}{\mathsf{inst}}
\newcommand{\Search}{\mathsf{searchpast}}

\newcommand{\derivlength}{\mathbb{DL}}
\newcommand{\linearderivlength}{\mathbb{LDL}}

\newcommand{\sll}{\mathbb{SLL}}
\newcommand{\nsll}{\mathbb{NSLL}}
\newcommand{\spl}{\mathbb{SP{\ell}}}

\newcommand{\nextI}{\mathsf{next}}

\newcommand{\cp}[1]{\mathbf{#1}}
\newcommand{\Ptime}{\cp{PTIME}} 
\newcommand{\classP}{\cp{P}}      

\newcommand{\NP}{\cp{NP}}
\newcommand{\FNP}{\cp{FNP}}
\newcommand{\FPtime}{\cp{FPTIME}}
\newcommand{\FP}{\cp{FP}}
\newcommand{\cPspace}{\cp{\sharp PSPACE}}
\newcommand{\Pspace}{\cp{PSPACE}}
\newcommand{\FPspace}{\cp{FPSPACE}}
\newcommand{\FPspaceN}{\cp{FPSPACE}}

\newcommand{\myFPspace}{\mathcal{F}_\cp{PSPACE}}

\newcommand{\suffix}{\textsf{suffix}}

\begin{document}
  
  \maketitle

\begin{abstract}
  This paper studies the expressive and computational power of
  discrete Ordinary Differential Equations (ODEs), a.k.a. (Ordinary) Difference Equations.  It presents a new
  framework using these equations   as a central tool for computation and
  algorithm design. We present the general theory of discrete ODEs for
  computation theory, we illustrate this with various examples of
  algorithms, and we provide several implicit characterizations of
  complexity and computability classes.

The proposed framework presents an original point of view on
complexity and computation classes. It unifies several constructions
that have been proposed for characterizing these classes including
classical approaches in implicit complexity using restricted recursion
schemes, as well as recent characterizations of computability and
complexity by classes of continuous ordinary differential equations.
It also helps understanding the relationships between analog
computations and classical discrete models of computation theory.


At a more technical point of view, this paper points out the
fundamental role of linear (discrete) ODEs
and classical ODE tools such as changes of variables to capture
computability and complexity measures, or as a tool for programming
many algorithms.


\end{abstract}





 
\paragraph{Note} The current article is a journal version
 of \cite{DBLP:conf/mfcs/BournezD19}. 
 
\section{Introduction}

Since the beginning of the foundations of computer science, the classification of the
difficulty of problems, with
various models of computation, 
either by their complexity or by their computability properties, is a
thriving field. 
Nowadays, classical 
computer science problems
also deal with continuous data coming from different areas and modeling
involves the use of tools like numerical analysis, probability theory
or differential equations. Thus new characterizations related to
theses fields have been proposed.
On a dual way,  the quest for
new types of computers recently led to revisit the power of some
models for analog machines based on differential equations, and to
compare them to modern digital models.
In both contexts, when discussing the related computability or
complexity issues, one has to
overcome the fact that today's (digital) computers are in essence
discrete machines while the objects under study are continuous and
naturally correspond to Ordinary Differential Equations (ODEs).

We consider here an original approach in
between the two worlds: discrete based computation with
difference equations.

ODEs  
appear to be a natural way of expressing properties and
are intensively used, in particular in applied science. 
The theory of classical (continuous) ODEs has an abundant literature (see
e.g. \cite{Arn78, BR89,CL55}) and is rather well
understood under many aspects. 
We are interested here in a discrete counterpart of
classical continuous ODEs: discrete ODEs, also known as difference equations.  Its associated derivative notion, called  \emph{finite differences},  has been widely studied in numerical optimization for function approximation \cite{gelfand1963calcul} and  in \emph{discrete calculus} 
\cite{graham1989concrete,gleich2005finite,izadi2009discrete,urldiscretecalculuslau} for  combinatorial   analysis (remark that similarities between discrete and
continuous statements have also been historically observed, under the
terminology of  \emph{umbral} or \emph{symbolic calculus} as early as in the $19$th century).  However, even if the underlying computational content of finite differences theory is clear and has been pointed out many times, no fundamental connections with algorithms and complexity  have been exhibited so far.
%
%
%

\shorter{
 Indeed, one widely used framework in which many systems are described
 by ODEs. They  
appear to be a natural way of expressing properties and
are intensively used, in particular in applied science. 
The theory of classical (continuous) ODEs is rather very well
understood, broadly studied and taught with a plethora of literature, see
e.g. \cite{Arn78, BR89, CL55}.
    %
\olivier{déjà amendé, mais probablement encore à amender le paragraphe qui suit.}
We are interested here in a discrete counterpart of
classical continuous ODEs: discrete ODEs, whose theory is sometimes
called \emph{discrete calculus}, \emph{finite differences} or
\emph{finite calculus}. 
While one can find some results 
usually mostly
motivated by very particular applications in fields like
combinatorics  in references like
\cite{graham1989concrete,gleich2005finite,izadi2009discrete,urldiscretecalculuslau}, the theory of discrete ODE is indeed surprisingly not
developed  as a whole in recent books. It turns out that a rather old
literature however exist, often motivated by similarities between discrete statements and
continuous statements that have been historically observed, under the
terminology of  \emph{umbral
  calculus}. However, as far as we know, this theory have never
considered statements from a complexity point of view, as we do in
the current article.
}

In this article, our goal is to demonstrate that discrete ODEs are a
very natural tool for algorithm design and to  prove that  complexity and computability notions  can be  elegantly and simply captured using
discrete ODEs. We illustrate this by providing a characterization of $\FPtime$, the class of polynomial time computable functions. 
To this aim,  we will also demonstrate  how some notions from the 
analog world such as linearity of differential equations or derivatives along some particular functions (i.e. changes of variables) are representative of a certain computational hardness and can be used to solve
efficiently some (classical, digital) problems.

As far as we know, this is the first time that computations with
discrete ODEs and their related complexity aspects are considered. 
By
contrast, complexity results have been recently obtained
about classical (continuous) ODEs for various classes of functions,
mostly in the framework of computable analysis. The hardness of
solving continuous ODEs has been intensively discussed: for example
\cite{Ko83} establishes some bases of the complexity aspects of ODEs
and more recent works like \cite{kawamura2009lipschitz} or
\cite{collins2008effectivesimpl} establish links between complexity
or effective aspects of such differential equations.
%
%
We believe that investigating the expressive power of discrete ODE, can help to better understand complexity of computation for both the discrete and continuous settings. 
Indeed, on the one hand, our results offers a new machine independent perspective on classical
discrete computations, i.e. computations that deal
with bits, words, or integers. 
And, on the other hand, it relates classical
(discrete) complexity classes to analog computations,
i.e. computations over the reals, as analog computation have been
related in various ways to continuous ordinary differential equations,
and as discrete ODEs provide clear hints
about their continuous counterparts. A mid-term goal of this line of research  is also to bring insights from complexity
theory to the problem of solving ODE (discrete and, hopefully
also, numerical). A descriptive approach such as the one initiated
in the paper could help classifying large classes of ODEs by their
computational hardness and bring some uniformity to methods of
this field.

\shorter{
	In this paper, we build a new bridge between two distinct areas of
	theoretical computer science. We prove that peculiar discrete ordinary
	differential equations describe recursion schemes from the implicit
	complexity theory. We believe the current work brings a clear new
	perspective to understand the links between all these approaches. In
	particular, it provides simple(r) explanations to various statements
	that have been established for analog computations. As an example, we
	believe our characterizations of the level of the \Greg{} hierarchy
	provide a simpler explanation of the reasons of why results in the
	spirit of those obtained in \cite{Cam01,CMC02} hold. In particular, it
	helps to clearly point out which aspects of the statements are related
	to continuous computations versus discrete computations.
}


\paragraph*{From restricted recursion scheme to discrete differential equations}
Recursion schemes constitutes a major approach of computability theory and to some extent of complexity theory.
A foundational result in that spirit is due to Cobham, who gave in
\cite{cob65} a characterization of functions computable in polynomial time through the notion of \textit{bounded recursion on notations} (BRN, for short) (see Section  \ref{sec:cobham} for a review). 
Later, notions such as safe recursion \cite{bs:impl} or
ramification (\cite{LM93,Lei94} have allowed syntactical
characterizations of polynomial time or other classes \cite{lm:pspace}
that do not require the use of an explicit bound but at the expense of
rather sophisticated schemes. These works have therefore been at the
origin of the very vivid field of \textit{implicit complexity} at the
interplay of logic and theory of programming.

A  discrete function can also be described by its (right discrete) derivative i.e.  the value 
$ \tu f(x+1,\tu y) - \tu f(x,\tu y)$. It is straightforward to rewrite
primitive recursion  in this setting, though one may have to cope with
possibly negative values. To capture more fine grained time and space
measures we come up in the proposed context of discrete ODEs with at
least two original concepts which are very natural in the continuous setting:
 \begin{itemize}
 	\item
          deriving along a function: when such a function is
          suitably chosen this allows to control the number of steps
          in the computation;
	\item
          linearity that permits to control object sizes.
\end{itemize}

By combining these two approaches, we provide a  characterization of $\FPtime$ that does not require to specify an
explicit bound in the recursion, in contrast to Cobham's work, nor to assign
a specific role or type to variables, in contrast to safe recursion or ramification. The characterization happens to be very simple from a syntactical point of view using only natural notions from the world of ODE.   

This characterization is also a first step to convince the reader that  deep connections between complexity and ODE solving do exist and that these connections are worth to be further studied.

\paragraph*{Related works on analog computations}
%
As many historical or even possibly futuristic analog machines are
naturally described by (continuous) ODEs, the quest of understanding
how the computational power of analog models compares to classical
digital ones have led to several results relating classical
complexity to various classes of continuous ODEs. In particular, a
series of papers has been devoted to study various classes of the
so-called $\R$-recursive functions, after their introduction in
\cite{Moo95b} as a theoretical model for computations over the reals.
At the complexity level, characterizations of complexity classes such
as $\Ptime$ and $\NP$ using $\R$-recursive algebra have been
obtained \cite{MC06}, motivated in particular by the idea of
transferring classical questions from complexity theory to the context
of real and complex analysis \cite{LCM07b, MC06,MC05Eatcs}. But this
has been done with the addition of limit schemata and with a rather
different settings. 

More recently, revisiting the model of General Purpose Analog Computer
of Claude Shannon, it is has been proved that polynomial differential
equations can be considered as a very simple and elegant model in
which computable functions over the reals and polynomial time
computable functions over the reals can be defined without any
reference to concepts from discrete computation theory 
\cite{JournalACM2017,TheseAmaury}. We believe the current work is a
substantial step to understand the underlying power and theory of such
analog models, by providing concepts, definitions and results relating
the two worlds.

Refer to
\cite{DBLP:journals/corr/BournezGP16} for an up to date
survey about various works on analog computations, in particular in a
computation theory perspective. 

%
%

%


\paragraph{How to read the paper} As this article is mixing considerations coming from various fields, we tried to write it as much as possible in a self-contained manner. In particular,
\begin{itemize}
\item  for a reader familiar with discrete differences, a.k.a. discrete ODEs, Section~\ref{sec:discrete differentiability} can be mostly skipped, as it is mainly reformulating some known results. The only point with respect to classical literature is that we often write $f^{\prime}$ instead of $\Delta f$ to help readers not familiar with that field to grasp the intuition of some of the statements with respect to their usual continuous counterpart. Our concept of falling exponential (Definition \ref{def:fallingexp}) is just a notation, and seems possibly non-standard, even if clearly inspired from the very standard falling power. 
\item A reader familiar will computability theory may mostly skip Section~\ref{sec:reminder}, as it is mostly only recalling some well-known notions and definition from computability theory. The only point is that we present all the various classes in some algebraic fashion, but all the presented statements are well-known in computability theory. We later observe in Section~\ref{sec:Computability and Discrete ODEs} that all these schemas can be seen as particular natural discrete ODEs schemas. Maybe the most original observation is about the role played by linear ODEs in that framework, and its relations with the \Greg{} hiearchy.
\item For a reader familiar with complexity theory, Subsection \ref{sec:cobham} is pointing out that not only computability classes but also complexity classes can be characterized algebraically. This might be less-known, but we are only reviewing here some basic results from so-called implicity complexity theory.
\end{itemize}
The truly original parts of this article are about the characterization of complexity classes, starting from Section \ref{sec:restrict}. Another clear originality is related to our discussion on various 
ways to program with discrete ODEs, with several examples: In particular discussions in Sections \ref{sec:f:programming with ODE}  and \ref{sec:prog}.

\paragraph*{Structure of the paper}
In Section~\ref{sec:reminder}, we review some basics of the literature of computability and complexity theory. 
In Section~\ref{sec:discrete differentiability} we provide a short introduction to discrete difference equations and some general basics definitions and results. 
In Section~\ref{sec:Computability and Discrete ODEs} we provide by an illustration, through examples, of the programming ability of discrete ODE.
We then provide formal definitions of discrete ODE schemas together with
characterizations of classical computability classes.  
From Section \ref{sec:restrict}, the focus is put on complexity theory. 
We introduce the notion  of length-ODE which is central, together with the notion  of (essentially) linear differential equation,  for the characterization
of $\FPtime$ (Section~\ref{sec:A characterization of polynomial
	time}).
 A conclusion is given in Section~\ref{sec:extension} where we also discuss some possible extensions of the results. 

\paragraph{Note}  A preliminary version coauthored with Sabrina Ouazzani is available on \url{https://arxiv.org/abs/1810.02241}.

\section{Computability and complexity: a quick introduction}
\label{sec:reminder}










\subsection{Computability theory and bounded schemes}

Computable functions, that is functions computable by Turing machines or equivalent models have  originally been seen through the prism of classical recursion theory. 
In this approach, instead of words, one deals with functions over integers, that is
to say with functions $\tu f: \N^p \to \N^d$ for some positive integers
$p,d$. 

It is well known that all main classes of classical recursion theory can then be
characterized as closures of a set of 
basic functions by a finite number of  basic
rules (operations) to build new functions: See
e.g. \cite{Rose,Odi92,clote2013boolean}. 

In that context, the smallest set of functions containing functions $f_{1}, f_{2}, \dots, f_{k}$ that is closed under operations $op_{1}$, $op_{2}$, \dots $op_{\ell}$ is often denoted as $$[f_{1}, f_{2}, \dots, f_{k}; op_{1}, op_{2},\dots,op_{\ell}].$$

For example, we have.

%

\begin{theorem}[Total computable functions]
The set $\mathcal{C}$ of total computable functions correspond to $$\mathcal{C}=[\zero, \projection{i}{p}, \sucs; composition, primitive~recursion, safe~minimization]:$$
That is to say, a  total function over the integers is computable if and only if it belongs
  to the smallest set of functions that contains constant function
  $\zero$, the projection functions $\projection{i}{p}$, the functions successor $\sucs$,
  that is closed under composition, primitive recursion and
  safe minimization.  \end{theorem}

In this statement, $\zero$, $\projection{i}{p}$ and $\sucs$ are
respectively the functions from $\N \to \N$, 
$\N^p \to \N$ and $\N \to \N$, defined as  $n \mapsto 0$,
$(n_1,\dots,n_p) \mapsto n_i$, and $n \mapsto n+1$. 
The above statement relies on two important notions, primitive recursion and safe minimization whose definitions are recalled below. 

\begin{definition}[Primitive recursion] \label{def:classique}
Given functions $g: \N^p \to \N$ and $h: \N^{p+2} \to \N$, function
$f=\REC(g,h)$ defined by primitive recursion from $g$ and $h$ is the function $\N^{p+1} \to \N$ satisfying
\begin{eqnarray*}
f(0,\tu y) &=& g(\tu y) \\
f(x+1,\tu y)&=&h(f(x,\tu y),x,\tu y). \\
\end{eqnarray*}
\end{definition}

\begin{definition}[(Safe) Minimization] \label{def:classiqued}
	Given function $g: \N^{p+1} \to \N$, such that for all $x$ there
	exists $\tu y$ with $g(x,\tu y)=0$,  function $f = \SMIN(g)$ defined
	by (safe) minimization
	from $g$ is the (total) function $\N^{p} \to \N$ satisfying $
	\SMIN(g) : \tu y \mapsto \min\{x; g(x,\tu y)=0\}$. 
	\end{definition}

This latter definition of safe
minimization is considered here instead of classical minimization as we focus in this
article only on total functions. 
Forgetting minimization, one obtains the following well-known class.

\begin{definition}[Primitive recursive functions]
We denote by $$\PR=[\zero, \projection{i}{p}, \sucs; composition, primitive~recursion]$$ the class of primitive recursive functions: 
A function over the integers is primitive recursive if and only if it belongs
  to the smallest set of functions that contains constant function
  $\zero$, the projection functions $\projection{i}{p}$, the functions successor $\sucs$,
  that is closed under composition and primitive recursion. 
\end{definition}

Primitive recursive functions have been stratified into various subclasses. 
We recall here the \Greg{} hierarchy in the rest of this subsection.

\begin{definition}[Bounded sum] \label{def:bsumclassique}
Given functions $g(\tu y): \N^p \to \N$,
\begin{itemize}
\item function $f=\BSUM(g): \N^{p+1} \to \N$ is defined by
  $f : (x,\tu y) \mapsto \sum_{z \le x} g(z,\tu y)$.
\item function
  $f=\BSUMs(g): \N^{p+1} \to \N$ is defined by
  $f : (x,\tu y) \mapsto \sum_{z < x} g(z,\tu y)$ for $x \neq 0$, and
  $0$ for $x=0$.
\end{itemize}

\end{definition}

\begin{definition}[Bounded product] \label{def:bprodclassique}
Given functions $g: \N^p \to \N$,
\begin{itemize}
\item function $f=\BPROD(g) : \N^{p+1} \to \N$ is defined by
$f: (x,\tu y) \mapsto \prod_{z \le x} g(z,\tu y)$.
\item function $f=\BPRODs(g)$ is defined by
$f: (x,\tu y) \mapsto \prod_{z < x} g(z,\tu y)$ for $x \neq 0$, and
$1$ for $x=0$.
\end{itemize}
\end{definition}

We have 
\begin{eqnarray*}
\BSUM(g)(x,\tu y) &=&\BSUMs(g)(x,\tu y)+g(x,\tu y) \\
\BPROD(g)(x,\tu y)&=&\BPRODs(g)(x,\tu y) \cdot g(x,\tu y).\\
\end{eqnarray*}


\begin{definition}[Elementary functions]
We denote by $$\Elem=[\zero,\projection{i}{p},\sucs, \plus, \myominus; composition, \BSUM, \BPROD]$$ the class of elementary functions:
A function over the integers is elementary if and only if
it belongs to the smallest set of functions that contains constant
function $\zero$, the projection functions $\projection{i}{p}$, the functions
successor $\sucs$, addition $\plus$, limited subtraction $\myominus: (n_1,n_2) \mapsto
max(0,n_1-n_2)$, and that is closed under composition, bounded sum
and
bounded product. 

\end{definition}

The class $\Elem$ contains many classical functions.
In particular:

\begin{lemma}[{\cite[Lemma 2.5, page 6]{Rose}}] $(x,y)\mapsto \lfloor
  x/y \rfloor$ is in $\Elem$. \label{lem:derose}
\end{lemma}
\begin{lemma}[{\cite{Rose}}] $(x,y) \mapsto x \cdot y$ is in $\Elem$.
\end{lemma}



%



The following normal form is also well-known. 

\begin{theorem}[Normal form for computable functions
  \cite{Kal43,Rose}] \label{normalform}
  Any total recursive function $f$ can be written as $f = g (
    \SMIN(h))$ for some elementary functions $g$ and $h$.
\end{theorem}

Consider the family of functions $\gE_n$  defined by induction
as follows. When $f$ is a function, $f^{[d]}$ denotes its $d$-th
iterate: $f^{[0]}(\vx)=x$, $f^{[d+1]}(\vx)=f(f^{[d]}(\vx))$:
\begin{eqnarray*}
\gE_0(x) &=& s(x) =  x+1, \\
\gE_1(x,y) &=& x+y, \\
\gE_2(x, y) &=& (x+1) \cdot (y+1),\\
\gE_3(x) &=& 2^x,\\
\gE_{n+1}(x)&=& \gE_n^{[x]}(1) \mbox{ for $n \geq 3$.}
\end{eqnarray*}


\begin{definition}[Bounded recursion] \label{def:limclassique}
Given functions $g(\tu y): \N^p \to \N$,  $h(f,x,\tu
y): \N^{p+2} \to \N$ and $i(x,\tu y): \N^{p+1} \to \N$, the function
$f=\limREC(g,h,i)$ defined by bounded recursion from $g$, $h$ and $i$  is defined as the function $\N^{p+1} \to \N$
verifying
\begin{eqnarray*}
f(0,\tu y) &=&  g(\tu y) \\
f(x+1,\tu y) &=&h(f(x,\tu y),x,\tu y) \\
\mbox{
under the condition that: } \\
f(x,\tu y) &\le& i(x, \tu y).
\end{eqnarray*}
\end{definition}

\begin{definition}[\Greg{} hierarchy (see \cite{Rose})]
  Class $$\E_{0}=[\zero,\projection{i}{p},\sucs;composition, \limREC]$$ denotes the class that contains the constant function
  $\zero$, the projection functions $\projection{i}{p}$, the successor
  function $\sucs$, and that is closed under composition and bounded
  recursion.

  For every $n \geq 1$, the class $$\E_{n}=[\zero,\projection{i}{p},\sucs,\gE_{n};composition, \limREC]$$ 
  
  is defined similarly except that
  functions max and $\gE_n$ are added to the list of initial functions.
\end{definition}

\begin{\propositions}[Odi92,Campagnolo,CamMooCos02]  \label{rqBsumbprod} Let $n
  \geq 3 $. 
$$\Gregn= [\zero,\projection{i}{p},\sucs, \plus, \myominus, \gE_{n}; composition, \BSUM, \BPROD].$$
\end{\propositions}

The above proposition means that closure under bounded recursion is
equivalent to using both closure under bounded sum and closure under
bounded product. Indeed, as explained in chapter $1$ of \cite{Rose}
(see Theorem 3.1 for details), bounded recursion can be expressed as a
minimization of bounded sums and bounded products, itself being
expressed as a bounded sum of bounded products.
The following facts are known:

\begin{\propositions}[\cite{Rose,Odi92,clote2013boolean}]
\begin{eqnarray*}
\Elem_3 &=& \Elem \subsetneq \PR 
\\ 
\Elem_{n} &\subsetneq& \Elem_{n+1} \mbox{ for $n \geq 3$} \\
\PR &=& \bigcup_{i} \Elem_i 
\end{eqnarray*}
\end{\propositions}

%
%

\subsection{Complexity theory and bounded schemes}
\label{sec:cobham}


We suppose the reader familiar with the basics of complexity theory (see e.g. \cite{Sip97}): Complexity theory is a finer theory whose aim is to discuss the resources such as time or space that are needed to compute a given function. In the context, of functions over the integers similar to the framework of previous discussion, the complexity of a function is measured in terms of the length (written in binary) of its arguments.

As usual, we denote by  $\Ptime$ (also denoted $\classP$ in the literature), resp. $\NP$, resp. 
$\Pspace$, the classes of decision problems decidable in deterministic polynomial time, resp. non deterministic polynomial time, resp polynomial space, on Turing machines. We denote by $\FPtime$ (or, shorter, $\FP$), resp. $\FPspaceN$ the classes of functions, $f:\N^k\rightarrow \N$ with $k\in \N$, computable in polynomial time, resp. polynomial space, on deterministic Turing machines. Note that while $\FPtime$ is closed by composition, it is not the case of $\FPspaceN$ since the size of the output can be exponentially larger than the size of the input.

It turns out that the main complexity classes have  also been characterized
algebraically, by restricted form of recursion scheme.
A foundational result in that spirit is due to Cobham, who gave in
\cite{cob65} a characterization of function computable in polynomial time. The idea is to consider schemes similar to primitive
recursion, but with restricting the number of
recursion steps.

Let $\succzero{.}$ and $\succun{.}$ be the successor functions defined by
$\succzero{x}=2.x$
and $\succun{x}=2.x+1$.

\begin{definition}[Bounded recursion on notations]
  A function $f$ is defined by bounded recursion scheme on notations
  from $g, h_0, h_1, k$, denoted by $f=\BRN(g,h_0,h_1,k)$, if 
\begin{eqnarray*}
 f(0, \tu y) &=& g(\tu y)\\
 f(\succzero{x}, \tu y) &=& h_0(f(x, \tu y) , x, \tu y) \mbox{ for $x \neq 0$} \\
f(\succun{x}, \tu y) &=& h_1(f(x, \tu y), x, \tu y) \\
\mbox{
under the condition that: } \\
  f(x, \tu y) &\le & k(x, \tu y) 
\end{eqnarray*}
for all $x,\tu y$. 
\end{definition}

Based on this scheme, Cobham proposed the following class of functions:  We write $\length{x}$ for the length (written in binary) of $x$. 

\begin{definition}[$\mathcal{F}_p$]
  Define $$\mathcal{F}_p=[\zero,\projection{i}{p},\succzero{x},\succun{x},\#; composition, \BRN]: $$ this is the smallest class of primitive
  recursive functions containing $\zero$, the projections $\projection{i}{p}$, the
  successor functions $\succzero{x}=2.x$ and $\succun{x}=2.x+1$, the function
  $\# $ defined by $x \# y = 2 ^{\length{x} \times \length{y}}$ and closed by
  composition and by bounded recursion scheme on notations.
\end{definition}

This class turns out to be a characterization of polynomial time:

\begin{theorem}[\cite{cob65}, see \cite{Clo95} for a proof]
  $$\mathcal{F}_p = \FPtime.$$
\end{theorem}

Cobham's result opened the way to various characterizations of
complexity classes, or various ways to 
control recursion schemes. This includes the famous
characterization of $\Ptime$ from Bellantoni and Cook
in~\cite{bs:impl} and by Leivant in~\cite{Lei94}. Refer to
\cite{Clo95,clote2013boolean} for monographies presenting a whole
series of results in that spirit. 


The task to capture $\FPspaceN$ is less easy since the principle of such characterizations is to use classes of functions closed by composition. However, for functions with a reasonable output size some characterizations have been obtained.
Let us denote by $\myFPspace$, the subclass of $\FPspaceN$ of functions of polynomial growth i.e. of functions $f:\N^k \rightarrow \N$, such that, for all $\tu x \in \N^k$, $\length{f(\tu x)}=O(\max_{1\leq i \leq k} \length{x_i})$. The following then holds:

\begin{theorem}[{\cite{thompson1972subrecursiveness},\cite[Theorem 6.3.16]{clote2013boolean}}]
$$\myFPspace= [\zero, \projection{i}{p}, \sucs, \#; composition,  \limREC].$$
That is to say, 
a function over the integers is in $\myFPspace$ if and only if it belongs
  to 
  the smallest set of functions that contains the constant function
  $\zero$, the projection functions $\projection{i}{p}$, the functions
  successor $\sucs$, function $\#$ 
  that is closed under composition and bounded recursion. 
\end{theorem}




\section{Discrete difference equations and discrete ODEs}\label{sec:discrete differentiability}


From now on, on this section, we review some basic notions of discrete calculus to
help intuition in the rest of the paper (refer to
\cite{jordan1965calculus,gelfand1963calcul} for a more complete review).
%
%
Discrete derivatives, also known as discrete differences, are usually intended to concern functions over the
integers of type $\tu f: \N^p \to \Z^d$, for some integers $p,d$, but the statements and concepts
considered in our discussions  are also valid more generally for functions
of type $\tu f: \Z^p \to \Z^d$ or even functions
$\tu f: \R^p \to \R^d$.
The basic idea is to consider the following concept of
derivative:

\begin{remark}
  We first discuss the case where $p=1$, i.e. functions
  $\tu f: \N \to \Z^d$. We will later on consider more general
  functions, with partial derivatives instead of derivatives.
\end{remark}

\begin{definition}[Discrete Derivative] The discrete derivative of
	$\tu f(x)$ is defined as $\Delta \tu f(x)= \tu f(x+1)-\tu
        f(x)$. We will also write $\tu f^\prime$ for
        $\Delta \tu f(x)$ to help readers not familiar with discrete differences to understand
	statements with respect to their classical continuous counterparts. 
	
\end{definition}


Several results from classical derivatives generalize
to the settings of discrete differences: this includes linearity of derivation $(a \cdot
f(x)+ b \cdot g(x))^\prime = a \cdot f^\prime(x) + b \cdot
g^\prime(x)$, formulas for products
and division such as
 $(f(x)\cdot g(x))^\prime =
 f^\prime(x)\cdot g(x+1)+f(x) \cdot g^\prime(x)= f(x+1)  g^\prime(x) +  f^\prime(x)  g(x)$. 
    Notice that, however, there is no simple equivalent of the chain rule. 
 
A fundamental concept is the following:

\begin{definition}[Discrete Integral]
	Given some function $\tu f(x)$, we write $\dint{a}{b}{\tu f(x)}{x}$
	as a synonym for $\dint{a}{b}{\tu f(x)}{x}=\sum_{x=a}^{x=b-1}
        \tu f(x)$ with the convention that it takes value $0$ when $a=b$ and
        $\dint{a}{b}{\tu f(x)}{x}=- \dint{b}{a}{\tu f(x)}{x}$ when $a>b$. 
\end{definition}

The telescope formula yields the so-called Fundamental Theorem of
Finite Calculus: 

\begin{theorem}[Fundamental Theorem of Finite Calculus]
	Let $\tu F(x)$ be some function.
	Then,
	$$\dint{a}{b}{\tu F^\prime(x)}{x}= \tu F(b)-\tu F(a).$$
      \end{theorem}

      As for classical functions, a given function has several
      primitives. These primitives are defined up to some additive
      constant.  Several techniques from the classical settings
      generalize to the discrete settings: this includes the technique
      of integration
      by parts. 

%
%

\shorter{
\ENONCEINTEGRALPARAMETER{\label{derivintegral}}

}

A classical concept in discrete calculus is the one of falling
	power defined as $$x^{\underline{m}}=x\cdot (x-1)\cdot (x-2)\cdots(x-(m-1)).$$
	This notion is  motivated by the fact that it satisfies a derivative formula $
        (x^{\underline{m}})^\prime  = m \cdot x^{\underline{m-1}}$  similar to the classical
        one for powers in the continuous setting.
In a similar spirit, we 
introduce the concept of falling exponential. 


\begin{definition}[Falling exponential] \label{def:fallingexp} 
	Given some function $\tu U(x)$, the expression $\tu U$ to the
	falling exponential $x$,
	denoted by $\fallingexp{\tu U(x)}$, stands
        for  \begin{eqnarray*}
                \fallingexp{\tu U(x)} &=&
                                                                                (1+ \tu U^\prime(x-1)) \cdots
                                        (1+ \tu U^\prime(1)) \cdot (1+ \tu U^\prime(0))   \\
                                         &=&
                   \prod_{t=0}^{t=x-1} (1+ \tu U^\prime(t)),
                    \end{eqnarray*}
	with the convention that $\prod_{0}^{0}=\prod_{0}^{-1}=\tu {id}$, where $\tu
        {id}$ is the identity (sometimes denoted  $1$ hereafter)
\end{definition}

This is motivated by the remarks that 
	$2^x=\fallingexp{x}$, and
        that the discrete
	derivative of a falling exponential is given by
	$$\left(\fallingexp{\tu U(x)}\right )^\prime = \tu U^\prime(x) \cdot
	\fallingexp{\tu U(x)}$$
	%
      for all $x \in \N$.
      
\begin{lemma}[Derivation of an integral with parameters]  \label{derivationintegral}
   Consider $$\tu F(x) = \dint{a(x)}{b(x)} {\tu f(x,t)}{t}.$$
   Then \begin{eqnarray*}
          \tu F'(x) &=& \dint{a(x)}{b(x)}{  \frac{\partial \tu f}{\partial
                        x} (x,t)}{t} 
                     + \dint{0}{-a^\prime(x)}{\tu f(x+1,a(x+1)+t)}{t} 
+ \dint{0}{b'(x)}{ \tu f(x+1,b(x)+t ) } {t}. 
\end{eqnarray*}

\noindent In particular, when $a(x)=a$ and $b(x)=b$ are constant functions, $\tu F'(x) = \dint{a}{b}{  \frac{\partial \tu f}{\partial
     x} (x,t)}{t},$
     and when $a(x)=a$ and $b(x)=x$,
     $\tu F'(x) = \dint{a}{x}{  \frac{\partial \tu f}{\partial
     x} (x,t)}{t} + \tu f(x+1,x)$.

\end{lemma}

 \begin{proof}
\begin{eqnarray*}
  \tu F'(x) &=& \tu F(x+1) - \tu F(x) \\
  &=& \sum_{t= a(x+1)}^{b(x+1) - 1} \tu f(x+1,t) -
                          \sum_{t= a(x)}^{b(x) - 1} \tu f(x,t)   \\
  &=& \sum_{t= a(x)}^{b(x)-1} \left( \tu  f(x+1,t) - \tu f(x,t)
      \right)    
     +   \sum_{t=a(x+1)} ^{t= a(x)-1} \tu f(x+1,t) 
                      + \sum_{t=b(x)}
      ^{b(x+1)-1} \tu f(x+1,t) \\
  &=& \sum_{t= a(x)}^{b(x)-1}  \frac{\partial \tu f}{\partial
      x} (x,t)  + \sum_{t=a(x+1)} ^{t= a(x)-1} \tu f(x+1,t) + \sum_{t=b(x)}
      ^{b(x+1)-1} \tu f(x+1,t) \\
  &=& \sum_{t= a(x)}^{b(x)-1}   \frac{\partial \tu f}{\partial
      x} (x,t)  + \sum^{t=-a(x+1)+a(x)-1} _{t=0} \tu f(x+1,a(x+1)+t)
  + \sum_{t=0}
      ^{b(x+1)-b(x)-1} \tu  f(x+1,b(x)+t).
\end{eqnarray*}
\end{proof}


We will focus in this article on discrete Ordinary Difference 
Equations (ODE) on functions with several variables, that is to say
for example on equations of the (possibly vectorial) form:
%
%
%
\begin{equation}\label{lodepv}
\dderiv{\tu f(x,\tu y)}{x}= \tu h(\tu f(x,\tu y),x,\tu y),
\end{equation}
 where $\dderiv{\tu f(x,\tu y)}{x}$ stands as expected for the derivative of
functions $\tu f(x,\tu y)$ considered as a function of $x$, when $\tu y$
is fixed i.e. 
\[
\dderiv{\tu f(x,\tu y)}{x} = \tu f(x+1,\tu y) - \tu f(x,\tu y).
\]
 When some initial value $\tu f(0,\tu y) = \tu g(\tu y)$ is added, this is called an \emph{Initial Value
  Problem (IVP)} or a \emph{Cauchy Problem}\shorter{, also called a
  \emph{Cauchy Problem}}. An IVP can always be put in integral form $$\tu
f(x_{0},\tu y) = \tu f(0,\tu y) + \dint{0}{x_{0}}{\tu h(\tu f(u,\tu y),u,\tu y)}{u}.$$
\shorter{That is to say, we are
given a problem of type:

\begin{oureqnarray}\label{cauchy}
\dderiv{\tu f(x,\tu y)}{x}&=&\tu h(\tu f (x,\tu y),x,\tu y )\\
\tu f(0,\tu y)&=& \tu g(\tu y )
\end{oureqnarray}
\noindent with functions $\tu g,\tu h$ of suitable dimensions and domains:
}


Our 
aim here is to discuss total functions whose domain and range
is either of the form $\mathcal{D}=\N$, $\Z$, or possibly a finite
product $\mathcal{D}= \mathcal{D}_1 \times \dots \times \mathcal{D}_k$
where each $\mathcal{D}_i=\N$, $\Z$.
By considering that $\N \subset \Z$, we  assume that the codomain is
always 
$\Z^d$ for some $d$. 
The concept of solution for such ODEs is  as expected: given $h: \Z^d \times \N \times \Z^{p} \to \Z$
(or $h: \Z^d \times \Z \times \Z^{p} \to \Z$),
a solution  
over $\mathcal{D}$ is a function $f: \mathcal{D}
\times \Z^{p} \to \Z^d$
that satisfies the equations for all $x,\tu y$.

We will only consider well-defined ODEs such as above in this
article (but variants with partially defined functions could be considered as well). 
Observe that an IVP of the form \eqref{lodepv} always admits a
(necessarily unique) solution over $\N$ since $f$ can be defined
inductively with 

\begin{equation}
\begin{split}
\tu f(0,\tu y)&=\tu g(\tu y)\\  \tu f(x+1,\tu
y)&= \tu f(x,\tu y)
+ \tu h(\tu f(x,\tu y),x, \tu y).\\
\end{split}
\end{equation}

\begin{remark}
Notice that this is not necessarily true over $\Z$: As an
example, consider 
$f^\prime(x) = - f(x) + 1$, $f(0) = 0$. 
By definition of $f^\prime(x)$, we must have $f(x+1)=1$ for all $x$,
but if $x = -1$, $f(0)=1 \neq 0$. 

\end{remark}

\begin{remark}[Sign function] \label{signn} It is very instructive to realize that the solution of the above
  IVP 
  over $\N$ is the sign $\signn{x}$ function defined by
  $$\signn{x}= \left\{ \begin{array}{ll} 1 & \mbox{ if $x>0$} \\ 0 & \mbox{otherwise.} \end{array}\right.$$
\end{remark}

\newcommand\approxint{\approx_\N}

\subsection{Discrete linear ODE}

Discrete \emph{linear} (also called \emph{Affine}) ODEs  i.e. discrete ordinary differential equations of the
form $\tu f^\prime(x)=\tu A(x) \cdot \tu f(x) + \tu B(x)$ will
play an important role
in what follows,

\begin{remark}
Recall that the solution of $ f^\prime(x)= a(x) f(x) + b(x)$ for
classical continuous derivatives turns out to be given by (usually
obtained using the method of variation of parameters
): 
$$	f(x) = f(0) e^{\int_0^{x} a(t)
		dt}  + \int_{0}^{x} b(u) e^{\int_u^{x} a(t)
		dt }
	du.$$
%
%
            This 
        generalizes with our definitions to discrete ODEs, and this works even vectorially as shown by the following lemma. 
        \end{remark}




\begin{lemma}[Solution of linear ODE
	] 
	For matrices $\tu A$ and vectors $\tu B$ and $\tu G$ with coefficients in $x$ and $\tu y$,
	the solution of equation $\tu f^\prime(x,\tu y)= \tu A(x,\tu y) \cdot \tu f(x,\tu y)
	+  \tu
	B (x,\tu y)$  with initial conditions $\tu f(0,\tu y)= \tu G(\tu y)$ is
	\begin{eqnarray}\label{soluces}
        \tu f(x,\tu y) 
        =
          \left( \fallingexp{\dint{0}{x}{\tu
                             A(t,\tu y)}{t}} \right) \cdot \tu G (\tu
                             y)  
          +
	\dint{0}{x}{ \left(
		\fallingexp{\dint{u+1}{x}{\tu A(t,\tu y)}{t}} \right) \cdot
              \tu B(u,\tu y)} {u}.
          \end{eqnarray}

        \end{lemma}

        \begin{remark} \label{rq:fun}
          Notice that this can be rewritten in the familiar sum and product notation as $$
\tu f(x,\tu y)=\sum_{u=-1}^{x-1}  \left(
\prod_{t=u+1}^{x-1} (1+\tu A(t,\tu y)) \right) \cdot  \tu B(u,\tu y)
$$ with the (not so usual) conventions that for any function $\kappa(\cdot)$,  $\prod_{x}^{x-1} \tu \kappa(x) = 1$ and $\tu
B(-1,\tu y)=\tu G(\tu y)$.
Such equivalent expressions both have a clear computational content. They can
be interpreted as an algorithm unrolling
the computation of $\tu f(x+1,\tu y)$ from the computation of  $\tu
f(x,\tu y), \tu f(x-1,\tu y), \ldots, \tu f(0,\tu y)$.  
        \end{remark}

Actually, we can even prove a generalization of the previous statement. Matrices $\tu A$ and $\tu B$ may indeed be some arbitrary function of $x$. In particular, we observe  that this work even if we assume $\tu A$, $\tu B$ to be some function of $\tu f$ and of some external function $\tu h$. The statement of the previous lemma is clearly a particular case of what follows.

\begin{lemma}[Solution of linear ODE
	] \label{def:solutionexplicitedeuxvariables}
	For matrices $\tu A$ and vectors $\tu B$ and $\tu G$,
	the solution of equation $\tu f^\prime(x,\tu y)= \tu A(\tu f(x,\tu y),\tu h(x, \tu y), x,\tu y) \cdot \tu f(x,\tu y)
	+  \tu
	B (\tu f(x,\tu y), \tu h(x, \tu y),  x,\tu y)$  with initial conditions $\tu f(0,\tu y)= \tu G(\tu y)$ is
	\begin{eqnarray*}\label{soluce}
        \tu f(x,\tu y)  &=&
          \left( \fallingexp{\dint{0}{x}{\tu
                             A(\tu f(t,\tu y),\tu h(t, \tu y), t,\tu y)}{t}} \right) \cdot \tu G (\tu
                             y)  \\
           &&
          +
	\dint{0}{x}{ \left(
		\fallingexp{\dint{u+1}{x}{\tu A(\tu f(t,\tu y),\tu h(t, \tu y), t,\tu y)}{t}} \right) \cdot
              \tu B(\tu f(u,\tu y),\tu h(u, \tu y), u,\tu y)} {u}.
          \end{eqnarray*}

        \end{lemma}
                
\begin{proof}
Denoting the right-hand side by $\tu {rhs}(x,\tu y)$, we  have 
$$\begin{array}{ccl} \bar {\tu {rhs}}^{\prime}(x,\tu y)&=& \tu A(\tu f(x,\tu y),\tu h(x, \tu y), x,\tu y)  \cdot  \left( \fallingexp{\dint{0}{x}{\tu
                             A(\tu f(t,\tu y),\tu h(t, \tu y),  t,\tu y)}{t}} \right)   \cdot \tu G (\tu
                             y) \\
                                     && +
                             \dint{0}{x}{ \left(
		\fallingexp{\dint{u+1}{x}{\tu A(\tu f(t,\tu y), \tu h(t, \tu y),  t,\tu y)}{t}} \right)^{\prime} \cdot
              \tu B(\tu f(u,\tu y), \tu h(u, \tu y), u,\tu y)} {u} \\ 
              && +
                             \left(
		\fallingexp{\dint{x+1}{x+1}{\tu A(\tu f(t,\tu y),\tu h(t, \tu y), t,\tu y)}{t}} \right) \cdot
              \tu B(\tu f(x,\tu y), \tu h(x, \tu y), x,\tu y) \\
              &=& \tu A(\tu f(x,\tu y),\tu h(x, \tu y), x,\tu y)  \cdot  \left( \fallingexp{\dint{0}{x}{\tu
                             A(\tu f(t,\tu y),\tu h(t, \tu y),  t,\tu y)}{t}} \right)   \cdot \tu G (\tu
                             y)  \\ && + \  \tu A(\tu f(x,\tu y), \tu h(x, \tu y), x,\tu y) \cdot \\
                             && \dint{0}{x}{ \left(
		\fallingexp{\dint{u+1}{x}{\tu A(\tu f(t,\tu y),\tu h(t, \tu y), t,\tu y)}{t}} \right)  \tu B(\tu f(u,\tu y), \tu h(u, \tu y), u,\tu y)} {u} \\
                    &&
              + \ \tu B(\tu f(x,\tu y),\tu h(x, \tu y), x,\tu y)
              \\
                            &=&  \tu A(\tu f(x,\tu y), \tu h(x, \tu y), x,\tu y) \cdot \tu {rhs}(x,\tu y) 	+  \tu
	B (\tu f(x,\tu y),\tu h(x, \tu y), x,\tu y)
	              \end{array} 
              $$
              where we have used linearity of derivation and definition of falling exponential for the first term, and derivation of an integral (Lemma \ref{derivationintegral}) providing the other terms to get the first equality, and then the definition of falling exponential.
              This proves the property by unicity of solutions of a discrete ODE, observing that $\bar {\tu {rhs}}(0,\tu y)=\tu G(\tu y)$.
\end{proof}

        \begin{remark} \label{rq:fund}
          Notice that this can still be rewritten  as $$
\tu f(x,\tu y)=\sum_{u=-1}^{x-1}  \left(
\prod_{t=u+1}^{x-1} (1+\tu A(\tu f(t,\tu y), \tu h(t, \tu y),  t,\tu y)) \right) \cdot  \tu B(\tu f(u,\tu y), \tu h(u, \tu y), u,\tu y)
.$$ Again, this can be interpreted as an algorithm unrolling
the computation of $\tu f(x+1,\tu y)$ from the computation of  $\tu
f(x,\tu y), \tu f(x-1,\tu y), \ldots, \tu f(0,\tu y)$ in a dynamic programming way.  The next section will build on that remark through some examples.

Note that it  could be that $\tu A(\tu f(x,\tu y),\tu h(x, \tu y),x,\tu y)=\tu f(x,\tu y)$, i.e. we would be considering $\tu f^{\prime}(x,\tu y) = \tu f(x,\tu y)^{2}$ that would not be expected to be called a linear ODE (but rather a quadratic or polynomial one). Later, when talking about computability and complexity issues, we will forbid such possibilities, in order to really focus on linear ODEs. But at that point, as the previous statements hold in full generality, we stated them as such.
 \end{remark}

\section{Computability and Discrete ODEs}\label{sec:Computability and Discrete ODEs}

We now go to the heart of the subject of the current paper, and we discuss discrete ODEs as a way to program. We will illustrate this with examples in the simpler context of computability, and we will then later go to complexity theory.

Indeed, the computational dimension of calculus of finite differences has been widely stressed in mathematical analysis. However, no  fundamental connection has been established with algorithmic and complexity. 
In this section, we show that several algorithms can actually be
naturally expressed as discrete ODEs.
{
\subsection{Programming with discrete ODEs: First examples}\label{sec:f:programming with ODE}

Notice first that the theory of discrete ODEs also provides very natural alternative
ways to compute various quantities. 
First observe that  discrete ODEs allow to express easily 
search functions:

\paragraph{Searching with discrete ODEs}

As an illustration, suppose that we want to compute  the minimum of a function 
$$\min f : x \mapsto min \{f(y) : 0 \le y \le x\}.$$
This is given by $F(x,x)$ where $F$
is solution of the discrete ODE 
\begin{eqnarray*} F(0,x)  &=&
         f(0); \\
\dderiv{F(t,x)}{t}  &=&
  H(F(t,x),f(x),t,x), 
\end{eqnarray*}
  where $$H(F,f,t,x)= \left\{ \begin{array}{ll} 0 &\mbox{ if $F<f$},\\ f-F &\mbox{ if $F \ge f$}. \end{array}\right.$$
  In integral form, we
have: 
\[
F(x,y) = f(0) + \dint{0}{x}{H(F(t,y),t,y)}{t}.
\]

Conversely such an integral above (equivalently
discrete ODE) can 
always be considered as a (recursive) algorithm:
compute the integral from its definition as a sum.
%
%
On this example, this corresponds basically to compute $F(x,x)$ recursively by \label{ex:page}
%
    $$F(t+1,x) 
                    =
              \logicalif{F(t,x) < f(x)}{F(t,x)}{ f(x)}. 
$$ 
%
%
%
   %
   where $\logicalif{a}{b}{c}$ is $b$ when $a$ is true, $c$ otherwise.
   
\begin{remark}
Note that this algorithm is not polynomial 
 in the length of its argument $x$,
as 
it takes
time $\mathcal{O}(x)$ to compute $\min f$.
Getting to polynomial algorithms will be at the heart of coming
discussions starting from Section \ref{sec:restrict}. 
\end{remark}

The fact that discrete ODEs provide very natural alternative
ways (possibly not efficient) to compute various quantities is very clear when
considering numeric functions such as $\tan$, $\sin$, etc. 

\paragraph{Computing $\tan$ with discrete ODEs, iterative algorithms}




As an illustration, suppose one wants to compute $\tan(x_0)$ for say $x_0=72$.
One way to do it is to observe that 
\begin{equation}\label{eq:tan} \tan(x)^\prime= \tan(1) \cdot 	(1+\tan(x) \tan(x+1)).
\end{equation}
From fundamental theorem of finite calculus 
we can hence write:
\begin{eqnarray}
\label{eq:tran}
  tan(x_0)&=&0 + \dint{0}{x_0}{\tan^{\prime}(x)}{x}  \label{eq:unt}
  \\
&=& 0 + \tan(1) \cdot 
  \dint{0}{x_0}{(1+\tan(x)\tan(x+1))}{x}.  \label{eq:untt}
\end{eqnarray}

Inspired from previous remarks, the  point  is that Equation
\eqref{eq:unt} can be interpreted as an algorithm: it provides a
way to compute $\tan(x_0)$ as an integral (or if you prefer as a
sum).

Thinking about what  this integral means, discrete ODE  \eqref{eq:tan},
also encoded by \eqref{eq:untt}, can also be interpreted as
$$\tan(x+1)-\tan(x)= \tan(1) \cdot  [1 + \tan(x) \tan(x+1)]$$ 
  that is to say
 $$\tan(x+1)= f(\tan(x))$$
where   $f(X)= \frac{X+\tan(1)}{1-\tan(1)X}$. 
Hence, this is suggesting  a way to compute $\tan(72)$ by a method close to
express that $tan(x_0)=f^{[x_0]}(0).$
That is to say Equations \eqref{eq:unt} and \eqref{eq:untt} can be
interpreted as providing a way to compute $\tan(72)$ using an
iterative algorithm: they basically encode some recursive way of
computing $\tan$.












Of course, a similar principle would hold for $\sin$, or $\cos$ using
discrete ODEs for these functions, and for many other functions starting
from expression of their derivative. 

%
%

\begin{remark}
Given $x_0$, (even if we put aside how to deal with involved real
quantities) a point is that computing $\tan(x_0)$ using this method can not be
considered as
polynomial time, as the (usual) convention is that
time complexity is measured in term of the length of $x_0$, and not on
$x_0$. 
\end{remark}

Again, this example is illustrative and is not claimed to be an efficient algorithm.
  Getting to discrete ODEs that would yield to a better complexity is at the
  heart of  the remaining of the article starting from Section \ref{sec:restrict}.   In particular, by playing with change of
 variable so that 
  the integral becomes
  computable in  polynomial time.

Before getting to this efficiency issues, we first consider  functions defined by discrete ODEs under the prism of
computability. It serves as an introductory illustration of  the close
relationship between discrete ODEs and recursive schemata before moving to efficient algorithms and complexity classes characterizations with a finer approach.


\subsection{Computability theory and discrete ODEs}
\label{sec:calc}

Through this section, we will see that the class of functions that can be programmed using discrete ODEs is actually (and precisely) the whole class of computable functions. We  state actually finer results, using various classes of computability theory reviewed in Section \ref{sec:reminder}.
This part  is clearly inspired by ideas from \cite{Cam01,cam:moo:fgc:00}, but adapted here for our framework of discrete
ODEs.
%



\begin{definition}[
 Discrete ODE schemata]  \label{def:dode}
Given $g: \N^p \to \N$ and
  $h: \Z\times \N^{p+1} \to \Z$, we say that $f$ 
is defined by discrete ODE solving from $g$ and $h$, denoted by
$f=\ODE(g,h)$,  if $f: \N^{p+1} \to \Z$
corresponds to the (necessarily unique) solution of  Initial Value
Problem

\begin{equation}\label{eq:cauchy}
\begin{split}
f(0,\tu y)&= g(\tu y ), \\
\dderiv{f(x,\tu y)}{x} &= h(f(x,\tu y),x,\tu  y).
\end{split}
\end{equation}
\end{definition}

\begin{remark}
To be more general, we could take $g:\N^p \to \Z$. However, this would
be of no use in the context of this paper. We could also consider
vectorial equations as in what follows which are very natural in the
context of ODEs, and we would get similar statements. 
\end{remark}

\shorter{Primitive recursion can
be reformulated as a
particular discrete ODE schemata.

\begin{theorem}[A discrete ODE characterization of primitive recursive
  functions]
 The set of primitive recursive functions $\PR$ is the intersection
 with $\N^\N$ of the smallest set of functions that contains the zero
 functions $\zero$, the projection functions $\projection{i}{p}$, the addition and subtraction functions $\plus$ and $\minus$, and that is closed under composition and discrete $\ODE$
 schemata.
\end{theorem}
}

It is clear that primitive recursion schemes can be reformulated as discrete ODE schemata.

\begin{theorem}
  \begin{eqnarray*}
  \PR &=& \N^\N \cap [\zero,\projection{i}{p}, \sucs; composition, ODE].
  \end{eqnarray*}
  \end{theorem}
  
The restriction to Linear ODEs is very natural, in particular as
this class of ODEs has a highly developed theory for the continuous
setting.  It is very instructive to realize that the class of
functions definable by Linear ODEs is exactly the well-known class of
elementary functions \cite{Kal43} as we will see now. 



\begin{definition}[Linear ODE schemata]  \label{def:lode3} \label{def:lode}
Given a vector $\tu G=(G_i)_{1 \le i \le k}$, 
matrix $\tu A=(A_{i,j})_{1 \le i,j \le k}$,
$\tu B= (B_i)_{1 \le i \le k}$ whose coefficients corresponds to functions
$g_i: \N^p \to \N^k$, and 
$a_{i,j}: \N^{p+1} \to \Z$ and $b_{i,j}: \N^{p+1} \to \Z$
respectively, we say that
$\tu f$ is obtained by linear ODE solving from $G,A$ and $B$, denoted
by $\tu f=\LI(\tu G,\tu A, \tu B)$,
if  $f: \N^{p+1} \to \Z^k$ corresponds to the (necessarily unique) solution of  Initial Value
Problem
\begin{equation}
\label{eq:lincauchy3}
\begin{split}
\tu f(0,\tu y)&= \tu G(\tu y ) \\
\dderiv{\tu f(x,\tu y)}{x} &= \tu A(x,\tu y) \cdot \tu f(x,\tu y) +
\tu B(x,\tu y).
\end{split}
\end{equation}
\end{definition}



Bounded sum and product, two of the very natural operations at the core of the definition of elementary computable functions, easily come as solutions of linear ODE.

\begin{lemma}[Bounded sum and product] \label{def:bsumproduct}
	Let  $g: \N^{p+1} \to \N$.
	 \begin{itemize}
	 	\item 
        Function $f=\BSUMs(g)$ is the unique solution of initial value problem
		\begin{oureqnarrayde}
		 f(0,\tu y)&=&0, \\
			\dderiv{f(x,\tu y)}{x} &=& g(x,\tu  y); \\
			\end{oureqnarrayde}
		 \item 
		Function $f=\BPRODs(g)$  is the unique solution of initial value problem
		\begin{oureqnarrayde}
		                  f(0,\tu y) &=&
                        1,\\
                \dderiv{f(x,\tu y)}{x}  &=&
                                f(x,\tu y)
                                                           \cdot
                                                           (g(x,\tu
                                                           y)-1).  
                  		\end{oureqnarrayde}
		
	 \end{itemize}
	\end{lemma}

One key observation behind the coming characterizations is the following:

\begin{lemma}[Elementary vs Linear ODEs]\label{elemVsLinODE}
  Consider $\tu G,\tu A$ and $\tu B$ as in Definition \ref{def:lode}. Then
   Let  $\tu f=\LI(\tu G,\tu A, \tu B)$. Then ${\tu f}$ is elementary when $\tu G,
    {\tu A}$ and ${\tu B}$ are.
\end{lemma}

\begin{proof}
We do the proof in the scalar case, writing
$a,b,g$ for $\tu A,\tu B, \tu G$. The
general (vectorial) case follows from similar arguments. 
By Lemma~\ref{def:solutionexplicitedeuxvariables}, it follows that:

\begin{eqnarray*}
  f(x,\tu y) &=& \left( \prod_{t=0} ^{t=x-1} (1+a(t,\tu y)) \right)
                 \cdot g(\tu y) \\
  && +
     b (x-1,\tu y) \\
  &&+ 
  \sum_{u=0}^{x-2} \left( \prod_{t=u+1}^{x-1} (1+a(t,\tu y)) \right)
     \cdot b(u,\tu y).
     \end{eqnarray*}

  Clearly, $ \prod_{t=0} ^{t=x-1} (1+a(t,\tu y))=\BPRODs(1+a(t,\tu y))(x,\tu y)$. Similarly, let

  \begin{eqnarray*}
    p(u,x,\tu y) &=_{def}& \prod_{t=u+1}^{x-1} (1+a(t,\tu y))\\
    &=&\frac{\BPRODs(1+a(t,\tu y))(x,\tu y)}{\BPRODs(1+a(t,\tu y))(u+1,\tu y)}.
\end{eqnarray*}

  As the function $(x,y)\mapsto \lfloor x/y \rfloor$ is elementary from
  Lemma \ref{lem:derose}, we get that $p(u,x,\tu y)$ is elementary. 
As multiplication is elementary, it follows that $$\sum_{u=0}^{x-2}
 p(u,x,\tu y)  b(u,\tu y)  = \BSUMs(p(u,x,\tu y) b(u,\tu y) )(x-2,\tu y)$$ is also
elementary, and  ${f}$ is elementary using closure by
composition and multiplication. 
\end{proof}



The proof of this theorem easily follows.

\begin{theorem}[A discrete ODE characterization of elementary
  functions] \label{th:elem}
  \begin{eqnarray*}
  \Elem &=& \N^\N \cap [\zero,\projection{i}{p}, \sucs,\plus,\minus; composition, \LI].
  \end{eqnarray*}  That is to say, the set of elementary functions $\Elem$ is the intersection
  with $\N^\N$ of the smallest set of functions
  that contains the zero functions $\zero$, the projection functions $\projection{i}{p}$,
  the successor function $\sucs$, addition
  $\plus$, subtraction $\minus$, and that is closed under composition and
  discrete linear ODE schemata 
  $\LI$.
\end{theorem}

 By adding suitable towers of exponential as basis functions, the above result can be generalized to 
 characterize the
various levels of the \Greg{} hierarchy.

\begin{\propositions}\label{pr:GrzeEtODE} Let $n
	\geq 3 $. 
  \begin{eqnarray*}
	\Gregn &=& \N^\N \cap [\zero,\projection{i}{p}, \sucs, \gE_{n}; composition, \LI].
	\end{eqnarray*}
  \end{\propositions}

%
%
We can also reexpress Kleene's minimization: 



\begin{theorem}[Discrete ODE computability and classical computability are equivalent] 
	A total function $f: \N^p \to \N$ is  total
	recursive iff  there
	exist some functions $h_1,h_2: \N^{p+1} \to \N^2$ in the smallest set of functions
	that contains the zero functions $\zero$, the projection functions $\projection{i}{p}$,
	the successor function $\sucs$, and that is closed under composition and
	discrete linear  $\ODE$ schemata  such that:
	for all $\tu y$, 
	 \begin{itemize}
	 	\item
                  there exists some $T=T(\tu y)$ with
		$h_2(T,\tu y) = 0$;
		 \item
                  $f(\tu y) = h_1 (T,\tu y)$ where $T$ is the smallest such $T$. 
	\end{itemize}
\end{theorem}

\section{Restricted recursion and integration schemes}
\label{sec:restrict}


As illustrated so far, discrete ODEs are convenient tools to define functions. From their very definition, such schemes come with an evaluation mechanism  that makes their solution function computable. However, if one is interested in realistically computable functions, such as polynomial time ones, we are still lacking ODE schemes that structurally guarantee that their solution can be efficiently computed. We focus on this aspect in the rest of the paper.

\subsection{Programming with discrete ODEs: Going to efficient algorithms}
\label{sec:prog}

The previous examples discussed in Section \ref{sec:f:programming with ODE} were not polynomial. We now want to go to efficiency issues.
To do so, for now, we suppose that composition of functions, constant
and  the following basic functions can be used freely as functions
from $\Z$ to $\Z$:

 \begin{itemize}
 	\item 
          arithmetic operations: $+$, $-$, $\times$;
         \item 
          $\length{x}$ returns the length of $|x|$ written in binary;
	 \item
          $\sign{x}: \Z \to \Z$ (respectively: $\signn{x}: \N \to \Z$)
          that takes value $1$ for $x>0$ and $0$ in the other case; 
 \end{itemize}
From these basic functions, for readability, one may define useful functions as
synonyms: 
 \begin{itemize}
 	\item 
          $\signcomp{x}$ stands for 
          $\signcomp{x}=(1-\sign{x})\times (1-\sign{-x})$: it takes value in $\{0,1\}$ and values $1$ iff
          $x=0$ for $x \in \Z$;
 \item 
  $\signcompn{x}$ stands for 
          $\signcompn{x}=1-\signn{x}$: it takes value in $\{0,1\}$ and values $1$ iff
          $x=0$  for $x \in \N$. 
 \item 
  $\cond{x}{y}{z}$ stands for $\cond{x}{y}{z}=z +
  \signcomp{x}\cdot (y-z)  $ and 
$\condn{x}{y}{z}$ stands for $\condn{x}{y}{z}=z +
\signcompn{x}\cdot (y-z)  $: They value $y$ when $x=0$ and $z$ otherwise.
%
\shorter{
We have for both versions (The point is that the first considers $x \in \Z$ while the second
assumes $x \in \N$): 
$
\cond{x}{y}{z}=\left\{\begin{array}{l}
y \mbox{ if } x=0\\
z \mbox{ otherwise }
                      \end{array}\right.
                  $ 
                  $
\condn{x}{y}{z}=\left\{\begin{array}{l}
y \mbox{ if } x=0\\
z \mbox{ otherwise }
\end{array}\right.
$
}
%
%
  $\cond{x<x'}{y}{z}$ will be a synonym for
  $\cond{\sign{x-x'+1}}{y}{z}$. Similarly, $\cond{x \ge x'}{y}{z}$ will be a synonym for
  $\cond{\sign{x-x'+1}}{z}{y}$ and $\cond{x=x'}{y}{z}$
  will be a synonym for $\cond{1-\signcomp{x-x'}}{y}{z}$.
\end{itemize}

\paragraph{Doing a change of variable}
We illustrate our discussion through an example, before getting to the general theory.
\begin{example}[Computing the integer part and divisions, going to
  Length-ODE] \label{ex:some}
\label{sec:racine}
%
Suppose  that we want to compute $$\lfloor \sqrt{x} \rfloor = \max \{ y \le x : y \cdot y \leq x \}$$ and $$\left\lfloor
\frac{x}{y} \right\rfloor = \max
\{z \le x : z \cdot y \le x \}.$$ It can be done by the following
general method. Let $f,h$ be some functions with $h$ being non decreasing. We compute $\fonction{some}_h$  with $\fonction{some}_h(x)=y$ s.t. $|f(x)-h(y)|$ is minimal.   
When $h(x)=x^2$ and $f(x)=x$, it holds that:
$$
\lfloor \sqrt{x} \rfloor = \textsf{if}(\fonction{some}_h(x)^2 \leq x,\fonction{some}_h(x),\fonction{some}_h(x)-1).
$$
The relation $\fonction{some}_h$ can be computed
(in non-polynomial time)
as a solution of an ODE similar to what we did to compute the minimum of a function.

However, there is a more efficient (polynomial time) way to do it based on what one 
usually does 
with classical ordinary differential equations: 
performing a
change of variable so that the search becomes logarithmic in $x$ through dichotomic search.
Indeed, we can write
$
\fonction{some}_h(x)  = G(\length{x},x)
$ 
for some function $G(t,x)$, that is a solution of

\shorter{defined by:
$
    G(0,x)
                =
          x;$ 
    $G(t+1,x) 
                    =
              \textsf{if}( h(G(t,x)) = f(x), G(t,x),$ 
         $\textsf{if}(h(G(t,x)) > f(x), G(t,x) - 2^{\length{x}-(t+1)},$
       $G(t,x) +2^{\length{x} - (t+1)}))
$ 
 hence, 
 }

 \begin{eqnarray*} G(0,x) &=&
          x;  \\
  \dderiv{G(t,x)}{t}  &=&
E(G(t,x),t,x)
\end{eqnarray*}
where $$E(G,t,x)= \left\{
 \begin{array}{lll}
 2^{\length{x}-t-1}   & \mbox{ whenever } & 
 h(G)>f(x),   \\
0  & \mbox{ whenever } &
h(G)=f(x)  \\
- 2^{\length{x}-t-1} & \mbox{ whenever } & 
h(G)<f(x). \\
 \end{array}
 \right.$$
   %
   \end{example}

The example above is indeed a discrete ODE whose solution is
converging fast (in polynomial time) to what we want.

Reformulating  
what we just did, 
we wrote $\fonction{some}_h(x)  = G(\length{x},x)$
using the solution of the above discrete ODE, i.e. the solution of 
$
G(T,y) = x + \dint{0}{T}{E(G(t,y),t,y)}{t}.
$
This provides a polynomial time algorithm to solve our problems 
 using
a new parameter $t=\length{x}$ logarithmic in $x$. Such techniques
will be at the heart of the coming results.


\shorter{
\begin{remark}
	This kind of constructions invites to write above kind of dynamics in
	the form 
	$$ \dderiv{\fonction{some}_h(x)}{\length{x}} = E(\fonction{some}_h(x),\length{x},x)$$
	
	This corresponds to what we will call a length-ODE. Length-ODE will provide
	ways to express fast computations. 
\end{remark}
}

\paragraph{Non-numeric examples}
Discrete ODEs turn out to be very natural in many other contexts, in particular non numerical ones,
where they would probably not be expected.

\begin{example}[Computing suffixes with discrete ODEs]
%
%
The suffix function, $\suffix(x,y)$ takes as input two integers $x$ and $y$ and outputs  the $\length{y}=t$ least significant bits of the binary decomposition of $x$.
We describe below a way to compute a suffix working over a
parameter $t$, that is logarithmic in $x$. 
Consider the following unusual algorithm that can be interpreted as a
fix-point definition of the function: $\suffix(x,y)=F(\length{x},y)$ where
%
\begin{eqnarray*}
    F(0,x) &=&
              x ; \\
      F(t+1,x)  &=&
                 \cond{\length{F(t,x)} = 1}{F(t,x)}{ 
                                                             F(t,x) - 2^{\length{F(t,x)} -1}}. 
 \end{eqnarray*}
%
%
%
This can be interpreted as a discrete ODE, whose solution is
converging fast again (i.e. in polynomial time) to what we want.
In other words, $\suffix(x,y)= F(\length{x},x)$ using the solution of
the IVP: \begin{eqnarray*} F(0,x)&=&x,\\ \dderiv{F(t,x)}{t}&=&\cond{\length{F(t,x)} = 1}{0}{- 2^{\length{F(t,x)} -1}}.\\
\end{eqnarray*}
 %

\shorter{
or of the 
the form of a length-ODE:
$$ \dderiv{\suffix(x,y)}{\length{x}} =
\cond{\length{\suffix(\length{x},u)} = 1}{0}{-
  2^{\length{\suffix(\length{x},y)} -1}}$$}

\end{example}

\shorter{
After this teaser, the rest of this article aims at discussing which
problems can be solved using discrete ordinary differential equations,
and with which complexity.  Before doing so, we need to review some
basic concepts and results from computation theory that we will be
needed in the rest of this article and that have been obtained at this date.
}

\subsection{Derivation along a function: the concept of $\lengt$-ODE}
\label{ssec:restrict}

In order to talk about complexity instead of simple computability, we need to
add some restrictions on the integration scheme. We  introduce the following variation on the notion of
derivation and consider derivation along some function $\lengt(x, \tu y)$. \shorter{As we
already saw on examples, and we will
see in the general case, this is motivated by the idea of being able to talk easily about
algorithms obtained by considering some suitable changes of variables. We will call this $\lengt$-ODEs. 
}


%
%
\shorter{
\begin{remark} Observe that this is necessary. Indeed, the solution of a polynomial ordinary differential
	equation (ODE) can grow very very fast.

	Indeed:
	\begin{eqnarray*}
		\left(\fallingexp{x}\right)^\prime &=& \fallingexp{x} \\
		\left(\fallingexp{\fallingexp{x}}\right)^\prime &=& \fallingexp{x} 
		\cdot \fallingexp{\fallingexp{x}} \\
		\left(\fallingexp{\fallingexp{\fallingexp{x}}}\right)^\prime &=&
		\fallingexp{x} \cdot
		\fallingexp{\fallingexp{x}} \cdot
		\fallingexp{\fallingexp{\fallingexp{x}}}
		\\
		&\vdots& 
	\end{eqnarray*}
	and so on,
	is solution of degree 2 polynomial ODE
	\begin{eqnarray*}
		y^\prime_1&=& y_1 \\
		y^\prime_2&=& y_1 \cdot y_2 \\
		y^\prime_3&=& y_2 \cdot y_3 \\
		&\vdots&
	\end{eqnarray*}
	with initial condition $$y_1(0)=y_2(0)=y_3(0)=\dots=1.$$
That means that if we consider a too general integration
        scheme, then we get such towers of exponentials. Clearly, such
        a function is not polynomial time computable, as only writing
        its value in binary cannot be done in polynomial time. 
      \end{remark}
   }
%
%



\begin{definition}[$\lengt$-ODE] Let $\lengt:\N^{p+1} \rightarrow \Z$ and $\tu h$ some function. We  write
	\begin{equation}\label{lode}
	\dderivL{\tu f(x,\tu y)}= \dderiv{\tu f(x,\tu y)}{\lengt(x,\tu
          y)} = \tu h(\tu f(x,\tu y),x,\tu y),
	\end{equation}

as a formal synonym for
$$ \tu f(x+1,\tu y)= \tu f(x,\tu y) + (\lengt(x+1,\tu y)-\lengt(x,\tu y)) \cdot
\tu h(\tu f(x,\tu y),x,\tu y).$$

When  $\lengt(x,\tu y)=\length{x}$, the length function, 
 we will call this special case a \emph{length-ODE}
\end{definition}
%
\begin{remark}
This is motivated by the fact that the latter expression is similar to
classical formula for classical continuous ODEs:
$$\frac{\delta f(x,\tu y )}{\delta x} = \frac{\delta
  \lengt (x,\tu y) }{\delta x} \cdot \frac{\delta f(x,\tu
  y)}{\delta \lengt(x, \tu y)}.$$ 
\end{remark}
This will allow us to simulate suitable changes of variables using
this analogy.
We
will talk about $\lengt$-IVP when some initial condition is added. 

\begin{example} It is easily seen that:
	
	\begin{itemize}
	\item function  $f: x  \mapsto
	2^{\length{x} }$  satisfies the equation $\dderivl{f(x)} = \length{x}' \cdot
	f(x)$

	For this, we have basically observed the
		fact that
		\begin{eqnarray*} (2^{\length{x}} ) ' &=& 2^{\length{x+1}} - 2^{\length{x}}  
								  = (2^{\length{x}'} -1)  \cdot
									  2^{\length{x}}  \\ &=& \length{x}' \cdot
									  2^{\length{x}}
		\end{eqnarray*}	

		 where in last line, we have used the fact that $2^{e}-1=e$ for $e \in
		\{0,1\}$. 
		
	\item function $f: x  \mapsto
	2^{\length{x}^2 }$ satisfies the equation $\dderivL{ f(x)}= \left(\length{x}^2\right)' \cdot
	f(x)$

	considering $\lengt(x)=\length{x}^2$.
	\end{itemize}\end{example}

	
	\begin{example} More generally
	$f(x,y)=2^{\length{x}\cdot \length{y}}$ is the solution of the
	following Length-IVP: 
	\begin{eqnarray*}
	f(0, y)&=&2^{\length{y}}  \\
	  \dderivl{f(x, y)} &=& 
							f(x,y)\cdot (2^{\length{y}}-1), 
	\end{eqnarray*}

	
	\noindent using $2^{e \cdot \length{y} }-1=e \cdot (2 ^{\length{y} }
	-1)$ for $e \in \{0,1\}$
	and
	\begin{eqnarray*}
	  \left(2^{\length{x}\cdot \length{y}} \right)' &=&
	\left(2^{\length{x}' \cdot \length{y}} - 1\right) \cdot 
														2^{\length{x}\cdot \length{y}}\\
	  &=& \length{x}' \cdot
	\left(2^{\length{y}}-1 \right) \cdot 2^{\length{x}\cdot \length{y}}.
	\end{eqnarray*}
\end{example}

\shorter{
  \begin{example}[Example \ref{ex:some} continued]
  The trick used in Example \ref{ex:some} can be read as using a new parameter
  $t=\length{x}$ logarithmic in $x$, using relation
  	$$ \dderiv{\fonction{some}_h(x)}{\length{x}} = E(\fonction{some}_h(x),\length{x},x)$$
      \end{example}
      }

%



\shorter{
  \begin{remark}
    More generally, 
$f(x,y)=2^{\length{x}\cdot \length{y}}$ is the solution of the
following Length-IVP: 
\begin{eqnarray*}
f(0, y)&=&2^{\length{y}}  \\
  \dderivl{f(x, y)} &=& 
                        f(x,y)\cdot (2^{\length{y}}-1), 
\end{eqnarray*}


\noindent using $2^{e \cdot \length{y} }-1=e \cdot (2 ^{\length{y} }
-1)$ for $e \in \{0,1\}$
and
\begin{eqnarray*}
  \left(2^{\length{x}\cdot \length{y}} \right)' &=&
\left(2^{\length{x}' \cdot \length{y}} - 1\right) \cdot 
                                                    2^{\length{x}\cdot \length{y}}\\
  &=& \length{x}' \cdot
\left(2^{\length{y}}-1 \right) \cdot 2^{\length{x}\cdot \length{y}} 
\end{eqnarray*}

\end{remark} 
} 

%

\subsection{Computing values for solutions of $\lengt$-ODE}
 
The main result of this part  
illustrates one key property of the $\lengt$-ODE 
scheme 
from a computational point of view: its dependence on the number of distinct values of function $\lengt$. So computing values $f(x)$ of a function $f$ solution of some $\lengt$-ODE system depends on the number of distinct values taken by $\lengt$ between $0$ and $x$.


\begin{definition}[$Jump_\lengt$]
Let $\lengt:\N^{p+1}\rightarrow \Z$ be some function.
Fixing $\tu y\in \N^p$, let
$$
Jump_\lengt(x,\tu y)=\{0 \le i \le x-1 | \lengt(i+1,\tu y) \neq \lengt(i,\tu
y)\}
$$
%
be the set of non-negative integers less than $x$ after which the value of
$\lengt$ changes.
\end{definition}

We also write:
 \begin{itemize}
 	\item 
          $J_{\lengt}=|Jump_\lengt(x,\tu y)|$ for its cardinality; 
	 \item  	
          $\alpha:[0..J_{\lengt}-1]\rightarrow Jump_\lengt(x,\tu y)$ for an
	increasing function enumerating the elements of $Jump_\lengt(x,\tu y)$:  If  $i_0 < i_1 < i_2 <
	\dots < i_{J_{\lengt}-1}$ denote
	all elements of $Jump_\lengt(x,\tu y)$, then
	$\alpha(j)=i_j\in Jump_\lengt(x,\tu y)$.
\end{itemize}

Technically, $\alpha$ should be written $\alpha_{\tu y}$, i.e. depends on $\tu y$. For simplicity of writing, we will write $\alpha$ without putting explicitly this dependency.}



\begin{lemma}
  \label{lem:fundamentalobservation} Let $\tu f: \N^{p+1}\rightarrow \Z^d$ and 
$\lengt:\N^{p+1}\rightarrow \Z$  be some functions.
Assume that \eqref{lode} holds. 
Then $\tu f(x,\tu y)$ is equal to:
	\begin{eqnarray*}
	\tu f(0,\tu y) 
+ \dint{0}{J_{\lengt}}{ 
          {\Delta 
          \lengt(\alpha(u),\tu y)}
        \cdot \tu h(\tu f(\alpha(u),\tu y),\alpha(u),\tu y)}{u}.\end{eqnarray*}
\end{lemma} 

\begin{proof}
	By definition, we have
	$$\tu f(x+1,\tu y) = \tu f(x,\tu y) + (\lengt(x+1,\tu y)-\lengt(x,\tu
	y)) \cdot \tu h(\tu f(x,\tu y),x,\tu y).$$
	Hence,
	\begin{itemize}
		\item as soon as $i \not\in Jump_\lengt(x,\tu y)$, then $\tu
		f(i+1,\tu y)=\tu f(i,\tu y)$, since
		$\lengt(i+1,\tu y)-\lengt(i,\tu y)=0$. In other words, $\Delta
		\tu f(i,\tu y)=0.$
		\item as soon as $i \in Jump_\lengt(x,\tu y)$, say $i=i_j$,
		then $$\Delta \tu f(i_j,\tu y) = (\lengt(i_j+1
		,\tu y)-\lengt(i_j,\tu y)) \cdot \tu h(\tu f(i_j,\tu y),i_j,\tu y)$$
		I.e.
		$
		\Delta \tu f(i_j,\tu y) =
		\Delta \lengt(i_j,\tu y) \cdot \tu h(\tu f(i_j,\tu y),i_j,\tu y) $.
	\end{itemize}
	
	Now  
	\begin{eqnarray*}
		\tu f(x,\tu y) &=& \tu f(0,\tu y) + \dint{0}{x}{\Delta
			{\tu f(t,\tu y)}}{t}  \\
		&=& \tu f(0,\tu y) +
		\sum_{t=0}^{x-1} \Delta {\tu f(t,\tu y)} \\
		&=& \tu f(0,\tu y) + \sum_{i_j
			\in Jump_\lengt(x,\tu y)} \Delta
		{\tu f(i_j,\tu y)} \\
		&=& \tu f(0,\tu y) +  \sum_{i_j \in Jump_\lengt(x,\tu y)} \Delta \lengt(i_j,\tu y) \cdot
		\tu h(\tu f(i_j,\tu y),i_j,\tu y)  \\
		&=& \tu f(0,\tu y) +  \sum_{j=0}^{J_{\lengt}-1} \Delta \lengt(\alpha(j),\tu y) \cdot
		\tu h(\tu f(\alpha(j),\tu y),\alpha(j),\tu y) \\
		&=&  \tu f(0,\tu y) + \dint{0}{J_{\lengt}}{\Delta
			\lengt(\alpha(u),\tu y) \cdot \tu h (\tu f(\alpha(u),\tu
			y),\alpha(u),\tu y)}{u}
	\end{eqnarray*}
	which corresponds to the expression. 
\end{proof}

\begin{remark}
Note that the above result would still hold with $\lengt$ and $h$
taking their images in $\R$.
\end{remark} 

 The above proof is based on (and illustrates) some fundamental
aspect of $\lengt$-ODE from their definition: for fixed $\tu y$, the value of $\tu f(x,\tu y)$ only
changes when the value of $\lengt(x,\tu y)$ changes.
 Under the previous hypotheses, there is then an alternative view to
understand the integral, by using a change of variable, and by
building a discrete ODE that mimics the computation of the
integral.

\shorter{
Basically, we are using the fact that we can consider some
parameter $t$ corresponding to $\lengt(x,\tu y)$.
\shorter{In the special case where
$\lengt(x, \tu y)$ is length $\length{x}$ (see coming discussion), this
parameter will be logarithmic in $x$.
}
Indeed:
}


%

\begin{lemma}[Alternative view] \label{fundobge}
Let 
$f: \N^{p+1}\rightarrow \Z^d$,
$\lengt:\N^{p+1}\rightarrow \Z$  be some functions and assume that \eqref{lode} holds.
Then $\tu f(x,\tu y)$ is given by 
$\tu f(x,\tu y)=  \tu {\bar F}(J_{\lengt}(x,\tu y),\tu y)$
where $\tu {\bar F}$ is the solution of initial value problem
\begin{eqnarray*}
\tu {\bar F}(0,\tu y)&=& \tu f(0,\tu y), \\
\dderiv{\tu {\bar F}(t,\tu y)}{t} &=&  {\Delta\lengt(\alpha(t),\tu y)}
                               \cdot
\tu h(\tu { \bar F}(t, \tu y),\alpha(t),\tu y).
\end{eqnarray*}

\end{lemma}

\begin{proof}
If we rewrite the previous integral as an ODE, we get that $\tu f(x,\tu y)=  \tu {\bar F}(J_{\lengt}(x,\tu y),\tu y)$
where $\tu {\bar F}$ is the solution of initial value problem
\begin{eqnarray*}
\tu {\bar F}(0,\tu y)&=& \tu f(0,\tu y), \\
\dderiv{\tu {\bar F}(t,\tu y)}{t} &=&  {\Delta\lengt(\alpha(t),\tu y)}
                               \cdot
\tu h(\tu { f}(\alpha(t), \tu y),\alpha(t),\tu y).
\end{eqnarray*}

But by induction, $\tu { f}(\alpha(t), \tu y)$ can be rewritten as  $\tu {\bar F}(J_{\lengt}(\alpha(t),\tu y),\tu y)=\tu {\bar F}(t,\tu y)$
as $J_{\lengt}(\alpha(t),\tu y)=t$.

\end{proof}

In the special case of a length ODE, that is where $\lengt(x,\tu y)=\length{x}$, it holds that
$J_{\lengt}=J_{\lengt}(x,\tu y)=|Jump_\lengt(x,\tu y)|=\length{x}-1$.
Hence, the preceding result can also be formulated as:

\begin{lemma}[Alternative view, case of Length ODEs] \label{fundob}
Let 
$f: \N^{p+1}\rightarrow \Z^d$,
$\lengt:\N^{p+1}\rightarrow \Z$  be some functions and assume that \eqref{lode} holds considering  $\lengt(x,\tu y)=\length{x}$.
Then $\tu f(x,\tu y)$ is given by 
$\tu f(x,\tu y)= \tu F(\length{x},\tu y)$
where $\tu F$ is the solution of initial value problem
\begin{eqnarray*}
\tu F(1,\tu y)&=& \tu f(0,\tu y), \\
\dderiv{\tu F(t,\tu y)}{t} &=&  
\tu h(\tu F(t, \tu y),2^{t}-1,\tu y).
\end{eqnarray*}

\end{lemma}

\begin{proof}
Consider $\tu F(t,\tu y) = \tu {\bar F}(t-1,\tu y)$, observing that 
$\dderiv{\tu F(t,\tu y)}{t} = \tu F(t+1,\tu y) - F(t,\tu y)= \tu {\bar F}(t,\tu y) - \tu {\bar F}(t-1,\tu y) = \dderiv{\tu {\bar F}(t-1,\tu y)}{t} 
= {\Delta\lengt(\alpha(t-1),\tu y)}
                               \cdot
\tu h(\tu { \bar F}(t-1, \tu y),\alpha(t-1),\tu y)
= 1 \cdot \tu h(\tu {  F}(t, \tu y),\alpha(t-1),\tu y).
$
Observing that $\alpha(t) = 2^{t+1}-1$, the assertion follows.
\end{proof}

	\begin{example}[Back to function $2^{\length{x}}$]
		To compute function $f: x \mapsto 2^{\length{x}}$, one can also remark that $f(x)=F(\length{x})$ where $F(t)=2^t$ is solution of IVP:  
		%
		$F'(t)= F(t)$, $F(0)=1$. 
		We have used a change of variable $t= \length{x}$. 
		
		\end{example}





\shorter{
\begin{example}
  The previous discussion about the complexity of computing  $x \mapsto
  2^{\length{x}}$ and $x \mapsto 2^{\length{x}^2}$ is a 
  concrete application of all these remarks.
\end{example}
                                     }


\shorter{
\begin{example}
  Let us consider an example, where $\lengt(x)$ is not
  $\length{x}$:  Recall function 
  $f: x \mapsto 2^{ \lfloor \sqrt{x} \rfloor }$ for which we
  established$
  \dderiv{f(x)}{\lengt(x)}
= 
f(x)
$ considering 
%
%
%
%
  where 
  $\lengt(x) = \lfloor \sqrt{x} \rfloor$, 








One may think that the number
$|Jump_\lengt(x)|$  of $\lengt$ 
  is hard to predict, but the point is to
  look at the method we devised to compute $ \lfloor \sqrt{x} \rfloor$
  in Example \ref{sec:racine}: It is basically expressing $\lfloor \sqrt{x}
  \rfloor$ as some function $G$ of $\fonction{some}_h(x)$: We wrote $$\lfloor \sqrt{x}
  \rfloor = G( \fonction{some}_h(x))$$ for some function $G$.
 Consequently, we could
also consider variable $\lengt_2(x)=\fonction{some}_h(x)$, and  see
from expressions that the
number of jumps $|Jump_\lengt(x)|$   of previous $\lengt$ is actually
related to the $|Jump_{\lengt_2}(x)|$  of
this new $\lengt_2(x)$.
We also have
$\dderiv{f(x)}{\lengt_2(x)} =                                                 f(x) $

Observing
that $\fonction{some}_h(x)$ is  in turn
computed in ``time'' $\length{x}$ using the method of Example
\ref{sec:racine},  the number of jumps
for all these $\lengt(x)$ is always polynomial, and we are guaranteed
that all these
expressions lead to fast (polynomial) algorithms.  
\end{example}
}




\shorter{
\begin{remark}
  This method clearly extends to more general functions: Generalizing
  the above reasoning, we can compute fast functions of type $x \mapsto g(
  \lfloor \sqrt{x} \rfloor)$ as soon as we have a fast ODE computing
  $g$. Similarly, $\lfloor \sqrt{x} \rfloor$ can be replaced by
  anything that can be computed fast basically using similar
  techniques.
  \end{remark}
}
 \shorter{
An important and natural case is the special case where $\lengt(x,\tu y)$ is the
usual one variable length function $\lengt(x,\tu y)=\length{x}$.
We will of course write $\dderivl{\tu f(x,\tu y)}
$ in that
case for $\dderivL{\tu f(x,\tu y)}$.
}

\shorter{2
\subsection{Length-ODEs}

An important and natural case is the special case where $\lengt(x,\tu y)$ is the
usual one variable length function $\lengt(x,\tu y)=\length{x}$.
We will of course write $\dderivl{\tu f(x,\tu y)}
$ in that
case for $\dderivL{\tu f(x,\tu y)}$. 

We can adapt the Lemma above to
this special case of a what we will call length-ODE. Namely:

\begin{corollary}[First view]~\label{corollary:fundamental observation}
	Let $\lengt:\N\rightarrow \N$ be defined by $\lengt(x)=\length{x}$ for all integer $x$ and $f$ satisfies the hypothesis of Lemma~\ref{lem:fundamental observation}. Then,

	$$\tu f(x,\tu y) =
\tu	f(0,\tu y) + \dint{0}{\length{x}}{\tu h(\tu f(2^u-1,\tu y),2^u-1,\tu y)}{u}$$

	\noindent Or, equivalently:

	$$\tu f(x,\tu y) = \tu f(0,\tu y) + \sum_{i=0}^{\length{x}-1}
        \tu h(\tu f(2^i-1,\tu y),2^i-1,\tu y)$$
\end{corollary}

\begin{proof} Immediate consequence of Lemma~\ref{lem:fundamental observation}. Function $\alpha$ is such that $\alpha(i)=2^i-1$.
\end{proof}


\begin{corollary}[Alternative view]

	Let $\lengt:\N\rightarrow \N$ be defined by $\lengt(x)=\length{x}$ for all
        integer $x$ and $f$ satisfies the hypothesis of
        Lemma~\ref{lem:fundamental observation}. 
Then    $\tu f(x,\tu y)$ is given by $\tu f(x, \tu y) =
F(\length{x},\tu y)$
\noindent where $\tu F$ is the solution of initial value problem


\[	\dderiv{\tu F(t,\tu y)}{t} = \tu h ( \tu F(t,\tu y ),t,\tu y) 
\mbox{ with }	\tu F(0,\tu y)=\tu f(0,\tu y)
\]

%
\end{corollary}

In other words, for $\lengt(x)=\length{x}$, this offers us also two ways to present a length-ODE for
a function $f(x,\tu y)$: either by considering equation of the type
of~\eqref{lode} or by considering $\tu f(x,\tu y) = \tu
F(\length{x},\tu y)$ where $\tu F$ given by an equation of the form:

\begin{equation}\label{lode2}
\dderiv{\tu F(t,\tu y)}{t} = \tu h(\tu F(t,\tu y),t,\tu y)
\end{equation}
with $\tu F(0,\tu y)=\tu f(0,\tu y)$. As before, the idea is that  $t$ is a parameter logarithmic in $x$, namely
$t=\length{x}$.

}

\newcommand\vectorp[2]{\left(\begin{array}{l}
#1 \\ #2 \\
\end{array} \right)}

\shorter{
\subsection{On linear length-recursion scheme}


The example of function $f(x)$ defined by:

	\[
f(0)=1\mbox{ and } \dderiv{f}{\lengt}(x) = f(x)\cdot f(x).
\]

\noindent illustrates that it is possible to derive fast growing functions by simple length-ODE. This is due to the presence of non linear  terms such as $f(x)^2$ in the right-hand side of an equation. To control the growth of functions defined by ODE one possibility is to restrict the way functions that appear in equations use their argument. While doing so, one challenge is then to design a restriction that is  flexible and powerful  enough to permit a natural and simple description of a rich set of functions, in particular polynomial time computable functions.
}

\shorter{
Our purpose now is to discuss which kind of problems can be solved
efficiently using similar techniques: it  turns out to be exactly all of $\FPtime$.
It will be made clear from the incoming discussion and results. 
}

\subsection{Linear $\lengt$-ODE}~\label{subsec: linear length ODE}

In this section we adapt the concept of linearity to $\lengt$-ODE and prove that the functions that are solutions of the so-called linear $\lengt$-ODE
are intrinsically polynomial time computable when the chosen $\lengt$ function do not change of values often, as it is the case for the length function $\length{x}$.

Let's first consider the following length-ODE:

\begin{eqnarray*}
f(0)&=&2\\
\dderivl{f}(x) &=& f(x)\cdot f(x) - f(x).
\end{eqnarray*}

The unique solution  $f$ of the equation is $f(x)=2^{2^{\length{x}}}$. 
This example illustrates that it is possible to derive fast growing functions by simple length-ODE. This is due to the presence of non linear  terms such as $f(x)^2$ in the right-hand side of an equation. To control the growth of functions defined by ODE one possibility is to restrict the way functions that appear in equations use their argument. While doing so, one challenge is then to design a restriction that is  flexible and powerful  enough to permit a natural and simple description of a rich rest of functions, in particular polynomial time computable functions.

\newcommand\polynomial{ \fonction{sg}-polynomial}

\begin{definition}
A \polynomial{}  expression $P(x_1,...,x_h)$ is an expression built-on
$+,-,\times$ (often denoted $\cdot$) and $\sign{}$ functions over a set of variables/terms $X=\{x_1,...,x_h\}$ and integer constants.
The degree $\deg(x,P)$ of a term $x\in X$ in $P$ is defined inductively as follows:
\begin{itemize}
	\item $\deg(x,x)=1$ and for  $x'\in X\cup \Z$ such that $x'\neq x$, $\deg(x,x')=0$;
	\item $\deg(x,P+Q)=\max \{\deg(x,P),\deg(x,Q)\}$;
\item $\deg(x,P\times Q)=\deg(x,P)+\deg(x,Q)$;
\item $\deg(x,\sign{P})=0$.
\end{itemize}
A \polynomial{}  expression $P$  is \emph{essentially constant} in
$x$ if $\degre{x,P}=0$. 
It is \emph{essentially linear} in
$x$ if $\degre{x,P}=1$ i.e. if there exist \polynomial{}  expression $P_1,P_2$ such that $P= Q_1\cdot x + Q_2$ and $\degre{x,Q_1}=\degre{x,Q_2}=0$. 

A vectorial function (resp. a matrix or a vector) is said to be a \polynomial{} expression if all its coordinates (resp. coefficients) are. 
It is said to be
\emph{essentially constant} (resp. \emph{essentially linear}) if all its coefficients are.
\end{definition}

Compared to the classical notion of degree in a polynomial expression,
here all subterms that are within the scope of a sign function contributes $0$ to the degree.

\begin{example} 
	Let us consider the following \polynomial{} expressions. 
   \begin{itemize}
   \item 
    The expression $P(x,y,z)=x\cdot \sign{(x^2-z)\cdot y} + y^3$
    is essentially linear in $x$, essentially constant in $z$ and not linear in
    $y$. 
     \item 
      The expression
    $P(x,2^{\length{y}},z)=\sign{x^2 - z}\cdot z^2 + 2^{\length{y}}$
    is essentially constant in $x$, essentially linear in
    $2^{\length{y}}$ (but not essentially constant) and not
    essentially linear in $z$. 
     \item 
      The expression:
%
    $ \cond{x}{y}{z}=z + \signcomp{x}\cdot (y-z)= z +
    (1-\sign{x})\cdot (1-\sign{-x})\cdot (y-z) $
    is essentially constant in $x$ and linear in $y$ and $z$.

	\item The following matrix is essentially linear in $z$ and $y$ and constant in $x$.  
 
	\[	A(x,y,z)=
		\begin{pmatrix}
			\sign{x - y} & \sign{x}\cdot y\\
			\sign{z^5 - x^3} & z
		\end{pmatrix}	
	\]
  \end{itemize}

\end{example}

We are now ready to define the following main concept of ODE.

\begin{definition}[linear $\lengt$-ODE]\label{def:linear lengt ODE}
Function $\tu f$ is linear $\lengt$-ODE definable (from $\tu u$, 
$\tu g$ and $\tu h$) if it corresponds to the
solution of $\lengt$-IVP
\begin{equation}\label{SPLode}
\begin{split}
	\tu f(0,\tu y) =& \tu g(\tu y), \\
	\dderivL{\tu f(x,\tu y)}=&   \tu u(\tu f(x,\tu y), \tu h(x,\tu y),
	x,\tu y) 	
\end{split}
\end{equation}
\noindent where $\tu u$ is \textit{essentially linear} in $\tu f(x, \tu y)$. 
When $\lengt(x,\tu y)=\length{x}$, such a system is called linear length-ODE.
\end{definition}

In other words, function $\tu f$ is linear $\lengt$-ODE definable if there exist ${\tu A} ( \tu f(x,\tu y), \tu h (x,\tu y),
x,\tu y)$ and ${\tu B} ( \tu f(x,\tu y), \tu h(x,\tu y),
x,\tu y)$ that are \polynomial{} expressions essentially constant in $\tu f(x,\tu y)$ such that 

\[
	\dderivL{\tu f(x,\tu y)} = {\tu A} ( \tu f(x,\tu y), \tu h(x,\tu y),
	x,\tu y) \cdot
	  \tu f(x,\tu y)
	  +   {\tu B} ( \tu f(x,\tu y), \tu h(x,\tu y),
	  x,\tu y).	
\]

	In all previous reasoning, we considered that a function over the integers is polynomial time
	computable if it is in the length of all its arguments, as this is the
	usual convention. When not
	explicitly stated, this is our convention.  
	As usual, we also say that some vectorial function (respectively:
	matrix) is polynomial time computable if all its components are.
	We need sometimes to consider also polynomial dependency directly
	in some of the variables and not on their length. This happens in the
	next fundamental lemma where we consider linear ODE but derivation on a variable $x$ (and not along a function $\lengt$).
	%
	%
	
	We use the sup norm for the length of vectors and matrices. Hence, given some matrix $\tu
	A=(A_{i,j})_{1 \le i \le n, 1 \le j \le m}$, 
	we set $\length{\tu A}=\max_{i,j}
	\length{A_{i,j}}$.

	

	\begin{lemma}[Fundamental observation] \label{fundamencore}
	Consider the ODE 
	\begin{equation} \label{eq:bc}
	\tu f^\prime(x,\tu y)=  {\tu A} ( \tu f(x,\tu y), \tu h(x,\tu y),
	x,\tu y) \cdot
	  \tu f(x,\tu y)
	  +   {\tu B} ( \tu f(x,\tu y), \tu h(x,\tu y),
	  x,\tu y).
	\end{equation}
	Assume:
	\begin{enumerate}
	\item The initial condition $\tu G(\tu y) = ^{def}
	  \tu f(0, \tu y)$, as well as $\tu h(x,\tu y)$ are polynomial time computable in $x$ and in the length of $\tu y$. 
	  \item ${\tu A} ( \tu f(x,\tu y), \tu h (x,\tu y),
	  x,\tu y)$ and ${\tu B} ( \tu f(x,\tu y), \tu h(x,\tu y),
	  x,\tu y)$ are \polynomial{} expressions essentially constant in $\tu f(x,\tu y)$.
	
	%
	\end{enumerate}
	
	Then, there exists a polynomial $p$ such that $\length{\tu f(x,\tu y)}\leq p(x,\length{\tu y})$ and $\tu f(x, \tu y)$ is polynomial time computable in $x$ and the length
	of $\tu y$. 
	\end{lemma}

\begin{proof}
We know by Remark~\ref{rq:fund} following Lemma \ref{def:solutionexplicitedeuxvariables} that we must have:

\[\tu f(x,\tu y)=\sum_{u=-1}^{x-1}  \left(
\prod_{t=u+1}^{x-1} (1+\tu A(\tu f(t,\tu y), \tu h(t, \tu y),  t,\tu y)) \right) \cdot  \tu B(\tu f(u,\tu y), \tu h(u, \tu y), u,\tu y)
.
\]
		
\noindent with the conventions that $\prod_{x}^{x-1} \tu \kappa(x) = 1$ and $\tu B( \cdot , -1,\tu y)=\tu G(\tu y)$. 		

This formula permits to evaluate $\tu f(x, \tu y)$ using a dynamic programming approach  in a number of arithmetic steps that is polynomial in $x$ and $\length{\tu y}$. Indeed: 
For any $-1 \le u \le x$, $\tu A(\tu f(u,\tu y),\tu h(u,\tu y), u,\tu y)$     and 
$\tu B( \tu f(u,\tu y), \tu h(u,\tu y), u,\tu y)$ are matrices whose coefficients are \polynomial{}. So assuming, by induction, that each $\tu f(u,\tu y)$, for $-1 \le u < x$, can be computed in a number of steps polynomial in $u$ and $\length{\tu y}$, so are the coefficients of  $\tu A( \tu f(u,\tu y), \tu h(u,\tu y), u, \tu y)$ and  $\tu B( \tu f(u,\tu y), \tu h(u,\tu y), u, \tu y)$ which involve finitely many arithmetic operations or sign operations from its inputs. Once this is done, computing $\tu f(x, \tu y)$  requires polynomially in $x$ many arithmetic operations: basically, once the values for $\tu A$ and $\tu B$ are known we have to sum up $x+1$ terms, each of them involving at most $x-1$ multiplications. 


We need now to prove that not only the arithmetic complexity is polynomial in $x$ and $\length{\tu y}$, but also the bit complexity. As the bit complexity of a sum, product, etc is polynomial in the size of its arguments, it is sufficient to show  that the growth rate of function $\tu f(x,\tu y)$ can be polynomially dominated. For this, recall that, for any $-1 \le u \le x$, coefficients of $\tu A(\tu f(u,\tu y),\tu h(u,\tu y), u,\tu y)$     and 
$\tu B( \tu f(u,\tu y), \tu h(u,\tu y), u,\tu y)$ are essentially constant in $\tu f(u,\tu y)$. Hence, the size of these coefficients do not depend on $\length{\tu f(u,\tu y)}$. Since, in addition, $\tu h$ is computable in polynomial time in $x$ and $\length{\tu y}$, there exists a polynomial $p_M$ such that:

\[
\max (\length{\tu A(\tu f(u,\tu y),\tu h(u,\tu y), u,\tu y)},\length{\tu B(\tu f(u,\tu y),\tu h(u,\tu y), u,\tu y)} \leq p_M(u, \length{\tu y}). 	
\]

It then holds that,

\[
	\length{\tu f(x+1, \tu y)}\leq p_M(x,\tu y) + \length{\tu f(x, \tu y)} +1 	
\]

It follows from an easy induction that we must have
$\length{\tu f(x, \tu y) } \le \length{G(\tu y)} + (x+1) \cdot p_{M}(x, \length{ \tu y})$ which gives the desired bound on the length of values for function $\tu f$. 
\end{proof}


The previous statements lead to the following:

\begin{lemma}[Intrinsic complexity of linear $\lengt$-ODE]~\label{lem:fundamentalobservationlinearlengthODE}
	Assume that $\tu f$ is the solution of \eqref{SPLode} and that functions $\tu u, \tu g, \tu h, \lengt$ and elements of $Jump_\lengt$ are computable in polynomial time. 
	Then, $\tu f$ 
	is computable in polynomial time.
\end{lemma}

\begin{proof} Let $\tu f$ be a solution of the linear $\lengt$-ODE \eqref{SPLode}

	%
	\noindent where $\tu u$ is \textit{essentially linear} in $\tu f(x, \tu y)$. From Lemma~\ref{fundobge}, $\tu f(x,\tu y)$ can also be given by $\tu f(x,\tu y)=  \tu {\bar F}(J_{\lengt}(x,\tu y),\tu y)$ where $\tu {\bar F}$ is the solution of initial value problem~\footnote{In the statement of Lemma~\ref{fundobge}, functions $\tu h$ can be considered as part of $\tu u$. They are presented seprately here since $\tu u$ will be considered as linear and $\tu h$ plays the role of auxiliary functions that may have been computed before}
	\begin{eqnarray*}
	\tu {\bar F}(0,\tu y)&=& \tu g(\tu y), \\
	\dderiv{\tu {\bar F}(t,\tu y)}{t} &=&  {\Delta\lengt(\alpha(t),\tu y)}
								   \cdot
	\tu u(\tu { \bar F}(t, \tu y),\tu h(\alpha(t),\tu y), \alpha(t),\tu y).
	\end{eqnarray*}

	Functions $\tu u$ are \polynomial{} expressions that are essentially linear in $\tu f(x,\tu y)$. So there exist matrices $\tu A$, $\tu B$ that are essentially constants in $\tu f(t,\tu y)$ such that 
	
	\[
		\dderivL{\tu f(x,\tu y)} = {\tu A} ( \tu f(x,\tu y), \tu h(x,\tu y),
		x,\tu y) \cdot
		  \tu f(x,\tu y)
		  +   {\tu B} ( \tu f(x,\tu y), \tu h(x,\tu y),
		  x,\tu y).		
	\]

	\noindent In other words, it holds
	
	$$
	\tu F’(t,\tu y)= \overline {\tu A} ( \tu F(t,\tu y),
	t,\tu y) \cdot
	\tu F(t,\tu y)
	+   \overline {\tu B} ( \tu F(t,\tu y),
	t,\tu y).
	$$
	by setting 
	\begin{eqnarray*}
		\overline {\tu A} ( \tu F(t,\tu y),
		t,\tu y) &=& \Delta\lengt(\alpha(t),\tu y)\cdot {\tu A} (\tu { \bar F}(t, \tu y),\tu h(\alpha(t),\tu y), \alpha(t),\tu y) \\
		\overline {\tu B} ( \tu F(t,\tu y),
		t,\tu y) &=& \Delta\lengt(\alpha(t),\tu y)\cdot {\tu B} (\tu { \bar F}(t, \tu y),\tu h(\alpha(t),\tu y), \alpha(t),\tu y)	
	\end{eqnarray*}
	
	The
	corresponding matrice $\overline {\tu A}$ and vector $\overline {\tu
		B}$ are essentially constant in $\tu F(t, \tu y)$. Also, 
	functions $\tu g, \tu h$ are computable in polynomial time, more precisely 
	polynomial in  $\length{x}$, hence in $t$, and $\length{y}$. Function $Jump_\lengt$ is polynomial time computable in $\length{x}$ and $\length{y}$. So given $t$, obtaining $\alpha(t)$ is immediate. This guarantees that all
	hypotheses of Lemma \ref{fundamencore} are true. We can then conclude remarking, again, that $t=\lengt(x)$. 	
\end{proof}

We are now ready to state the main result of this subsection as a corollary of the previous results.

\begin{corollary}\label{cor: linear length ode are in ptime}
	Assume that $\tu f$ is the solution of a linear length-ODE from polynomial time computable functions $\tu u, \tu g$ and $\tu h$.
	Then, $\tu f$ 
	is computable in polynomial time.	
\end{corollary}

\shorter{
\begin{lemma}[Fundamental Observation for linear $\lengt$-ODE]~\label{lem:fundamentalobservationlinearlengthODE}
	Assume that $\tu f(x,\tu y)$ is solution of \eqref{SPLode}. 
        Then $\tu f(x, \tu y)$ 
	can be computed in polynomial time 
	under the following conditions:
	\begin{enumerate}
		\item $\tu f(0, \tu y)= \tu g(\tu
                  y)$ is computable in polynomial-time. 
		\item\label{fund obs cond 2}  function $\tu h$  is computable in polynomial
		time. 
		\item \label{fund obs cond 3} Functions $\lengt$ and $Jump_\lengt$ are computable in polynomial time. 
\end{enumerate}
\end{lemma}
}

 \section{A characterization of polynomial time}
\label{sec:A characterization of polynomial time}

The objective of this section is to provide a characterization by ODE schemes of the function computable in polynomial time.
Before proving this result, we briefly describe the computation model that will be used in proofs.

\shorter{
\begin{lemma}[Fundamental Observation for linear $\lengt$-ODE]~\label{lem:fundamentalobservationlinearlengthODE}
	Assume that $\tu f(x,\tu y)$ is solution of \eqref{SPLode}. 
        Then $\tu f(x, \tu y)$ 
	can be computed in polynomial time 
	under the following conditions:
	\begin{enumerate}
		\item $\tu f(0, \tu y)= \tu g(\tu
                  y)$ is computable in polynomial-time. 
		\item\label{fund obs cond 2}  function $\tu h$  is computable in polynomial
		time. 
		\item \label{fund obs cond 3} Functions $\lengt$ and $Jump_\lengt$ are computable in polynomial time. 
\end{enumerate}
\end{lemma}
}

\subsection{Register machines}

%

A register machine program (a.k.a. \textsf{goto} program) is a finite sequence of ordered labeled instructions acting on a finite set of registers of one of the following type:
 \begin{itemize}
 	\item
          increment the $j$th register $R_j$ by the value of $i$th
	register $R_i$ and go the next instruction:
	$
	R_j:=R_j+R_i$;
	 \item 
          decrement the $j$th register $R_j$ by the value of $i$th
	register $R_i$ and go the next instruction: $
	R_j:=R_j-R_i$;
	 \item
          set  the $j$th register $R_j$ to integer $a$, for $a \in
	\{0,1\}$ and go the next instruction: $
	R_j:=a$;
	 \item 
          if register $j$ is greater or equal to $0$, go to instruction $p$ else
          go to next instruction: 
	$
	\mathsf{if}~ R_j \ge  0 \ \textsf{goto} \ p$;
	 \item
          halt the program:
	$
	\mathsf{halt}
	$
\end{itemize}

In the following, since coping with negative numbers on classical
models of computation can be done through simple encodings, we will
restrict ourselves to non-negative numbers. 

\begin{definition} Let  $t:\N \rightarrow \N$.
	A function $f:\N^p\rightarrow \Z$ is computable in time  $t$
	by a register machine $M$ with $k$ registers if:
	 \begin{itemize}
	 	\item 
                  when starting in initial configuration with registers $R_1,\dots, R_{\min(p,k)}$ set to $x_1,\dots,x_{\min(p,k)}$ and  all other registers to $0$ and
		\item starting on the first instruction (of label $0$);
	 \end{itemize}
	%
        machine $M$ ends its computation after at most
	$t(\length{\tu x})$ instructions where $\length{\tu x}=\length{x_1}+\cdots+\length{x_p}$ and with register $R_0$ containing $f(x_1,\dots,x_p)$.
	A function is computable in polynomial time by $M$ if there exists
	$c\in \N$ such that $t(\length{\tu x})\leq \length{\tu x}^c$ for all
	$\tu x=(x_1,...,x_p)$.
\end{definition}

The definition of register machines might look rudimentary however, 
the following is easy (but tedious) to prove for any reasonable encoding of integers by Turing machines.

\begin{theorem}
	A function $f$ from $\N^p\rightarrow \Z$ is computable in polynomial time on Turing machines iff it is computable in polynomial time on register machines.
\end{theorem}

\subsection{A characterization of polynomial time}

The results  of Section~\ref{subsec: linear length ODE}    show that functions defined by linear length-ODE
from functions computable in polynomial time, are indeed polynomial time. We are now ready to prove a kind of reciprocal result. For this, we will introduce a recursion scheme based on solving linear differential equations.

%
%
%
%


\begin{remark}
Since the functions we define take their values in $\N$ and have output
in $\Z$, composition is an issue. Instead of considering restrictions
of these functions with output in $\N$ (which is always possible, even
by syntactically expressible constraints), we simply admit that
composition may not be defined in some cases. In other words, we consider that composition is a partial operator.
\end{remark}



\begin{definition}[Linear Derivation on Length~\footnote{In the conference version of this paper (\cite{DBLP:conf/mfcs/BournezD19}), the class was called $\derivlength$. It seems more appropriate to emphasize on the linearity also in the name of the class}, $\linearderivlength$.]
Let $$\linearderivlength = [\mathbf{0},\mathbf{1},\projection{i}{p}, \length{x}, \plus, \minus, \times, \sign{x} \ ; composition, linear~length~ODE].$$
	That is to say,  $\linearderivlength$ is the smallest subset of  functions,
	 that contains   $\mathbf{0}$, $\mathbf{1}$, projections
         $\projection{i}{p}$,  the length function  $\length{x}$,
         the addition function $x \plus y$, the subtraction function $x \minus y$, the multiplication function $x\times y$ (often denoted $x\cdot y$), the sign function $\sign{x}$
	and closed under composition (when defined)  and linear length-ODE
        scheme. \end{definition}


\begin{remark}
	As our proofs show, the definition of $\linearderivlength$ would remain the same by considering closure under any kind of $\lengt$-ODE with $\lengt$ satisfying the hypothesis of Lemma~\ref{lem:fundamentalobservationlinearlengthODE}.  
\end{remark}



\begin{example}
A number of natural functions are in $\linearderivlength$. \shorter{The following result is immediate by inspection of the example from Section~\ref{sec:programming with ODE} and~\ref{sec:restrict}.}
%
Functions $2^{\length{x}}$, $2^{\length{x}\cdot \length{y}}$, $\cond{x}{y}{z}$, $\suffix(x,y)$, 
	$\lfloor \sqrt{x}\rfloor $, $\lfloor \frac{x}{y}\rfloor$, $2^{\lfloor \sqrt{x}\rfloor}$
	all belong to $\linearderivlength$ by some linear length-ODE left as an exercice.
\end{example}

We can now state the main complexity result of this paper.



\begin{theorem}\label{th:ptime characterization 2}
	$\linearderivlength= \FPtime$
      \end{theorem}

\begin{proof}
	The inclusion $\linearderivlength \subseteq \FPtime$ is a consequence of  Corollary~\ref{cor: linear length ode are in ptime} and on the fact that arithmetic operations that are allowed can be computed in polynomial time and that $\FPtime$ is closed under composition of functions.
	
	We now prove that  $\FPtime\subseteq \linearderivlength$.
	Let $f:\N^p\longrightarrow \N$ be computable in polynomial time and $M$ a $k$ registers machine that compute $f$ in time $\length{\tu x}^c$ for some $c\in \N$.
	We first describe the computation of $M$ by simultaneous recursion scheme on length for functions  $R_0(t,\tu x), ..., R_k(t,\tu x)$  and $\inst(t,\tu x)$ that give, respectively, the values of each register and the label of the current instruction at time $\length{t}$.

	We start with an informal description of the characterization. Initializations of the functions are given by:
	$R_0(0,\tu x)=0, R_1(0,\tu x)=x_1$, \dots, $R_p(0,\tu x)=x_p$, $R_{p+1}(0,\tu x)=\cdots = R_k(0,\tu x)=0$ et $\inst(0,\tu x)=0$.
	Let $m\in \N$ be the number of instructions of $M$ and let $l\leq m$. Recall that, for a function $f$, $\dderiv{f}{L}(t,\tu x)$ represents a manner to describe $f(t+1,\tu x)$ from $f(t,\tu x)$ when $L(t+1)=L(t)+1$.  We denote by, $\nextI_l^{I}$, $\nextI_l^{h}$, $h\leq k$, the evolution of the  instruction function and of register $R_h$ after applying instruction $l$ at any such instant $t$. They are defined as follows:
	
	\begin{itemize}
		\item If instruction of label $l$ if of the type $R_j:=R_j+R_i$, then:
		
		\begin{itemize}
			\item $\nextI_l^{I}=1$ since $\inst(t+1,\tu x)=\inst(t,\tu x)+1$
			\item $\nextI_l^{j}=R_i(t,\tu x)$ since $R_j(t+1,\tu x)=R_j(t,\tu x)+R_i(t,\tu x) $
			\item $\nextI_l^{h}=0$ since $R_h(t,\tu x)$ does not change for  $h\neq j$.
		\end{itemize}
		
		\item If instruction of label $l$ if of the type $R_j:=R_j-R_i$, then:
		
		\begin{itemize}
			\item $\nextI_l^{I}=1$ since $\inst(t+1,\tu x)=\inst(t,\tu x)+1$
			\item $\nextI_l^{j}=-R_i(t,\tu x)$ since $R_j(t+1,\tu x)=R_j(t,\tu x) -R_i(t,\tu x)  $
			\item $\nextI_l^{h}=0$ since $R_h(t,\tu x)$ does not change for  $h\neq j$.
		\end{itemize}

		\item If instruction of label $l$ if of the type $R_j:=a$, for
		$a \in \{0,1\}$ then:
		
		\begin{itemize}
			\item $\nextI_l^{I}=1$ since $\inst(t+1,\tu x)=\inst(t,\tu x)+1$
			\item $\nextI_l^{j}=a-R_j(t,\tu x)$ since $R_j(t+1,\tu x)=a  $
			\item $\nextI_l^{h}=0$ since $R_h(t,\tu x)$ does not change for  $h\neq j$.
		\end{itemize}

		\item If instruction of label $l$ if of the type $\mathsf{if}$  $R_j \ge 0 \ \textsf{goto} \ p$, then:
		
		\begin{itemize}
			\item $\nextI_l^{I}=\cond{R_j(t,\tu x) \ge 0}{p-\inst(t,\tu x)}{1}$ since, in case $R_j(t,\tu x)\ge 0$ instruction number goes from $\inst(t,\tu x)$ to $p$.
			\item $\nextI_l^{h}=0$.
		\end{itemize}

		\item If instruction of label $l$ if of the type $\mathbf{Halt}$, then:
		
		\begin{itemize}
			\item $\nextI_l^{I}=0$ since the machine stays in the same instruction when halting
			\item $\nextI_l^{h}=0$.
		\end{itemize}
	\end{itemize}

	The definition of function $\inst$ by derivation on length is now given by (we use a more readable "by case" presentation):

	\[
	\dderivl{\inst}(t,\tu x)= \case
	\left\{\begin{array}{l}
	\inst(t,\tu x)=1 \quad \nextI_1^I \\
	\inst(t,\tu x)=2 \quad \nextI_2^I\\
	\vdots \\
	\inst(t,\tu x)=m \quad \nextI_m^I.\\
	\end{array}
	\right.
	\]
	
	Expanded as an arithmetic expression, this give: 
	
	\[
	\dderivl{\inst}
	(t,\tu x)= \sum_{l=0}^m \big(\prod_{i=0}^{l-1} \sign{\inst(t,\tu x)- i}\big)\cdot \signcomp{\inst(t,\tu x) - l}\cdot \nextI_l^I.
	\]
	
	Note that each $\nextI_l^I$ is an expression in terms of $\inst(t,\tu x)$ and, in some cases, in $\sign{R_j(t,\tu x)}$, too (for a conditional statement).
	Similarly, for each $j\leq k$:
	
	\[
	\dderivl{R_j}
	(t,\tu x)= \sum_{l=0}^m \big(\prod_{i=0}^{l-1} \sign{\inst(t,\tu x)- i}\big)\cdot \signcomp{\inst(t,\tu x)- l}\cdot \nextI_l^j.
	\]

	It is easily seen that, in each of these expressions above, there is
	at most one occurrence of $\inst(t,\tu x)$ and $R_j(t,\tu x)$ that is
	not under the scope of an essentially constant function (i.e. the sign
	function). Hence, the expressions are of the prescribed form i.e. linear.
	
	
	We know  $M$ works in time $\length{\tu x}^c$ for some fixed
	$c\in\N$. Both functions $\length{\tu x}=\length{x_1}+ \ldots
+	\length{x_p}$ and $B(\tu x)=2^{\length{\tu x}\cdot \length{\tu x}}$ are in $\linearderivlength$. It is easily seen that : $\length{\tu x}^c\leq B^{(c)}(\length{\tu x}))$ where $B^{(c)}$ is the $c$-fold composition of function $B$. We can conclude by setting $f(\tu x)=R_0(B^{(c)}(\max(\tu x)),\tu x)$.  
\end{proof}

\shorter{The following normal form theorem can also be obtained }
\shorter{
It is clear by inspecting the proof of Theorem~\ref{th:ptime characterization 2} that polynomial time computations can be characterized by first setting up a differential system using only the basic operation to describe the content of registers over time, then built the polynomial bound (as a function from the input) and finally substitute this bound as maximal value of the time variable. We can formalize this situation below to obtain a normal form of the characterization.
}

%
%


\begin{definition}[Normal linear $\lengt$-ODE (N$\lengt$-ODE)]~\label{def:system of SLL
    ODE} Function $\tu f 
  $ is definable by a normal linear  $\lengt$-ODE if it corresponds to the
  solution of $\lengt$-IVP
\begin{eqnarray*}\label{SPLodeb}
\tu f(0,\tu y)&=&\tu v(\tu y) \\
\dderivL{\tu f(x,\tu y)}&=&   \tu u(\tu f(x,\tu y), x,\tu y), 
\end{eqnarray*}
\noindent where $\tu u$ is \textit{essentially linear} in $\tu f(x,
\tu y)$ and $\tu v$ is either the identity, a projection or a constant function. 

\end{definition}

From the proof of Theorem~\ref{th:ptime characterization 2} the result
below can be easily obtained. It expresses that composition needs to
be used only once as exemplified in the above definition.

\begin{theorem}[Alternative characterization of $\FPtime$] \label{th:pspace}
  A function $\tu f:\N^p\to \Z$ is in $\FPtime$ iff
  $\tu f(\tu y) = \tu g(\length{\tu y}^c,\tu y)$ for some integer $c$ and
  some $\tu g: \N^{p+1}\to \Z^{k}$ solution of a normal linear
  length-ODE  
$$
 	\dderiv{\tu g(x,\tu y)}{\length{x}} =   \tu u(\tu g(x,\tu y), x,\tu y).
  $$ 

      \end{theorem}

\begin{remark}
  From similar arguments, $\FPtime$ also
  corresponds to the class defined as $\linearderivlength$ but where linear length-ODE
  scheme is replaced by normal linear length-ODE scheme: this forbids a function already defined with some
  ODE scheme to be used into some other ODE scheme.
\end{remark}


%
%
%
%
%

      \section{Discussions and further works}
      \label{sec:extension}
        
        Our aim in this article was to give the basis
      of a presentation of complexity theory based on discrete
      Ordinary Differential Equations and their basic properties. We
      demonstrated the particular role played by affine ordinary
      differential equations in complexity theory, as well as the
      concept of derivation along some particular function
      (i.e. change of variable) to guarantee a low complexity.

Previous ideas used here for $\FPtime$ opens the way to provide a characterization of other
complexity classes: This includes the possibility of characterizing non-deterministic polynomial time, using an existential quantification, or the 
class $\cp{P}_{[0,1]}$ of functions computable in polynomial time over the
reals in the sense of computable analysis, or more general classes of
classical discrete complexity theory   such as 
$\FPspace$.       For the clarity of the current
      exposition, as this would require to introduce other types of schemata of
      ODEs, we leave this characterization (and improvement of the
      corresponding schemata to ``the most natural and powerful form'') for future work, but we
      believe the current article basically provides the key  types of
      arguments 
      to conceive that this is indeed possible. 

More generally, it is also very instructive to revisit classical algorithmic under
this viewpoint, and for example one may realize that even inside class
$\Ptime$, the Master Theorem (see e.g. \cite[Theorem
4.1]{Cormen20093rd} for a formal statement) can be basically read as a result
on (the growth of) a particular class of discrete time length
ODEs. Several recursive algorithms can then be reexpressed as particular
discrete ODEs of specific type.

\shorter{ by considering random access machines (RAM) instead of register
      machines (see Appendix \ref{sec:ram} for definitions) with
      specific instructions sets.
      Depending on the set of basic
      operations allowed in the RAM model, polynomial time computation
      relates to very different complexity classes as witnessed by the
      following statements (see formal statement and proof of Theorem \ref{theorem:ram for
        pspace2} in appendix): 
%
	\begin{enumerate}
		\item A function $f:\N^k\to \Z$ is computable in polynomial time, i.e. is in $\FPtime$, iff it is computable in polynomial time on a $\{+,-\}$-RAM with unit cost.
		\item A function $f:\N^k\to \Z$ is computable in $\FPspace$ iff it is computable in polynomial time on a $\{+,-, \times, \div\}$-RAM with unit cost.
	\end{enumerate}

Second item follows from the following arguments:  It as been
proved in \cite{Galota:2005bo}, that a function $f$ is in $\FPspace$
iff it is the difference of two functions $f_1,f_2:\N\to \N$ in
$\cPspace$, the class of functions that counts the number of accepting
computations of a non deterministic polynomial space Turing machine.
It follows from the result from \cite{Bertoni:1981fs}, that a
function is computable in polynomial time on a
$\{+,\moins, \times, \div\}$-RAM if and only if it belongs to
$\cPspace$. 

Using random access machines (RAM) instead of register
      machines,  $\FPspace$ can be shown to correspond to functions of
      type $f(\tu y) = g_1(h(\tu y), \tu y)$ where $\tu g$ is defined
      as a specific class of polynomial length-ODE with substitutions,
      and conversely. 
}

\shorter{
\section{A characterization of polynomial space}
\label{sec:A characterization of polynomial space}

\olivier{Des modifs}

\olivier{Note cryptique pour moi: on veut couvrir le cas: 
  \begin{eqnarray*} 
   \frac{ \partial \Search}{\partial \length{t'}} (t',t,\tu x) &=& \tu F(
          \Search (t',t,\tu x), \tu {State} (t', \tu x), A(t,\tu x),
                                                                   t',
                                                                   \tu
                                                                   x)
    \\
 \frac{ \partial \tu{State} }{\partial \length{t}} (t,\tu x) &=& \tu G(
          \tu{State} (t, \tu x),  \Search (t, t, \tu x),
                                                                 t, \tu x)
    \\
    \end{eqnarray*}
}

\begin{definition}[System of Polynomial Length-ODE with variable substitution]
  Functions $\tu f = f_1,...,f_p$ with, for each $i\leq p$,
    $f_i:\N^{k_i+1}\rightarrow \Z$ 
  are definable by a System of
        Polynomial $\lengthnotation$-ODE  with variable substitution
        ($\spl$-ODE) if they are  solutions of a system of equations
        of the following form, for each $i\leq p$
%
        \olivier{Je remplace:
	\begin{equation}\label{SPLodeb}
	\dderiv{f_i(x,\tu y_i)}{\length{x}}=  P_i(f_{i_1}(x,\tu s_{i}^{i_1}), \dots, f_{i_q}(x,\tu s_{i}^{i_q}), x,\tu y)
      \end{equation}
      par:}
    	\begin{equation}\label{SPLodeb}
          \dderiv{f_i(x,\tu y_i, \tu y)}{\length{x}}=  P_i(f_i(x,\tu
          y_i,\tu y), 
          f_{i_1}(x,\tu s_{i}^{i_1},\tu y), \dots, f_{i_q}(x,\tu
          s_{i}^{i_q},\tu y), x,\tu y)
      \end{equation}
\olivier{Je remplace: 
  \noindent with $f_i(0, \tu y_i) = g_i(\tu y_i)$ and where,
  par:}
 \noindent with $f_i(0, \tu y_i,\tu y ) = g_i(\tu y_i, \tu y )$ and where,
	\begin{itemize}
        \item
           \olivier{Je remplace:
             $i_1,...,i_q\in \{1,...,p\}$
             par}
           $i_1,...,i_q\in \{1,...,i-1,i+1, \dots, p\}$
		\item  $g_i, P_i$ are \polynomial{} expressions
		\item for each $\alpha\in \{i_1,...,i_q\}$, $\tu
                  s_{i}^{\alpha}\in \{x,\tu y_i\}^{k_{\alpha}}$ i.e.
                  $\tu s_{i}^{\alpha}$ is a substitution of variables
                  for function $f_{\alpha}$ by variables for function
                  $f_i$.
                \end{itemize}
      \end{definition}

     \olivier{Faut-il adapter la notation? Pas nécessairement, mais peut-être?.}


\olivier{Enlevé: pas super pertinent.

	\begin{example}
	The following system is a polynomial length-ODE :
		\[
	f(0)=1\mbox{ and } \dderiv{f}{\length{x}}
        = f(x)\cdot f(x).
	\]

      \end{example}

    }

	\begin{definition}[$\spl$]
		A function $f:\N^k\to \Z$ is in $\spl$ if there exists $g_1,...,g_p$ with $g_i:\N^{k_i+1}\to \Z$, for $i=1,...,p$ (with $k_1=k$) and $h:\N^k\to \N$ in $\sll$ (or, equivalently $\FPtime$) such that:
		\begin{itemize}
			\item $g_1,...,g_p$ are solutions of a $\spl$-ODE
   	
			\item for all $\tu y\in \N^k$, $f(\tu y) = g_1(h(\tu y), \tu y)$.
		\end{itemize}
	\end{definition}

\arnaud{Ancienne version : 

\begin{definition}[System of Polynomial $L$-ODE with substitution (SPL-ODE)]
	Functions $\tu f = f_1,...,f_p:\N^{k+1}\rightarrow \Z$  are definable by a System of Polynomial $L$-ODE  with polynomial substitution if they are  solutions of a system of equations of the following form, for each $i\leq p$:

	\begin{equation}\label{SPLodeb}
	\dderiv{f_i(x,\tu y)}{L} =  P_i(f_1(x,\tu s_{i1}(x,\tu y)), \dots, f_p(x,\tu s_{ip}(x,\tu y)), x,\tu y)
	\end{equation}

	\noindent with $f_i(0, \tu y) = g_i(\tu y)$ and where, for each $i\leq p$

	\begin{itemize}
	\item  $g_i, P_i$ are \polynomial{} expressions
		\item for each $j\leq p$, $\tu s_{ij}(x,\tu y)\in  \N[x,\tu y]^k$ i.e. is a vector of polynomials in $x$, $\tu y$ with integer coefficients.
	\end{itemize}
	\end{definition}

We denote such a system above by:

\begin{equation}\label{SPLodebb}
\dderiv{\tu f(x, \tu y)}{L} =  \tu P(\tu f(x,\tu s(x,\tu y)), x,\tu y)
\end{equation}

\begin{example}
The following system is a polynomial length-ODE :
	\[
f(0)=1\mbox{ and } \dderiv{f}{\lengt}(x) = f(x)\cdot f(x).
\]

\end{example}

Since composition between function will not be allowed, or only in a very restricted form, to capture polynomial space, polynomial substitution has been introduced to gain in flexibility. However,  simple variable substitution is enough as it  will be shown in proofs.

%
%
%

%

\begin{definition}[$\spl$]
	A function $f:\N^k\to \Z$ is in $\spl$ if there exists $g_1,...,g_p:\N^{k+1}\to \Z$ and $h:\N^k\to \N$ in $\sll$ (or, equivalently $\FPtime$) such that:
	\begin{itemize}
		\item $g_1,...,g_p$ are solutions of a SP$\lengt$-ODE
		\[
		\dderiv{\tu g(x, \tu y)}{L} =  \tu P(\tu g(x,\tu s(x, \tu y)), x,\tu y)
		\]

		\item for all $\tu y\in \N^k$, $f(\tu y) = g_1(h(\tu y), \tu y)$.
	\end{itemize}
\end{definition}

%

}

\begin{theorem}\label{th:car of pspace}  A function is in $\FPspace$ if and only if it belongs to $\spl$.
\end{theorem}
}

\vspace{-0.2cm}

\section*{Acknowledgements}

We would like to thank Sabrina Ouazzani for many scientific
discussions about the results in this article. 


\bibliographystyle{siamplain}


\end{document}